\documentclass[10pt, twocolumn]{IEEEtran}

\usepackage{graphicx,epsfig}
\usepackage[noadjust]{cite}
\usepackage{mcite}
\usepackage{amsfonts,helvet}
\usepackage{fancyhdr}
\usepackage{threeparttable}
\usepackage{epsf,epsfig}
\usepackage{amsthm}
\usepackage{amsmath}
\usepackage{siunitx}
\usepackage{amssymb}
\usepackage{dsfont}
\usepackage{subfigure}
\usepackage{color}
\usepackage{algorithm}
\usepackage{algpseudocode}
\usepackage{algcompatible}
\usepackage{enumerate}
\usepackage{cancel}
\usepackage[section]{placeins}
\usepackage{graphicx}
\usepackage{hyperref}

\newtheorem{theorem}{Theorem}

\newtheorem{corollary}{Corollary}

\newtheorem{lemma}{Lemma}

\newtheorem{remark}{Remark}

\begin{document}

\title{Cooperative Base Station Coloring for Pair-wise Multi-Cell Coordination}

\author{Jeonghun~Park,
 Namyoon~Lee,
 and~Robert W. Heath Jr.
\thanks{J. Park and R. W. Heath Jr. are with the Wireless Networking and Communication Group (WNCG), Department of Electrical and Computer Engineering, 
The University of Texas at Austin, TX 78701, USA. (E-mail: $\left\{\right.$jeonghun, rheath$\left\}\right.$@utexas.edu) 

N. Lee is with Intel Labs, 2200 Mission College Blvd, Santa Clara, CA 95054, USA (Email:namyoon.lee@intel.com)

This research was supported by a gift from Huawei Technologies Co. Ltd.}}

\maketitle \setcounter{page}{1} 


\begin{abstract}
This paper proposes a method for designing base station (BS) clusters and cluster patterns for pair-wise BS coordination.
The key idea is that each BS cluster is formed by using the 2nd-order Voronoi region, and the BS clusters are assigned to a specific cluster pattern by using edge-coloring for a graph drawn by Delaunay triangulation.
The main advantage of the proposed method is that the BS selection conflict problem is prevented, while users are guaranteed to communicate with their two closest BSs in any irregular BS topology.
With the proposed coordination method, analytical expressions for the rate distribution and the ergodic spectral efficiency are derived as a function of relevant system parameters in a fixed irregular network model. 
In a random network model with a homogeneous Poisson point process,
a lower bound on the ergodic spectral efficiency is characterized. 
Through system level simulations, the performance of the proposed method is compared with that of conventional coordination methods: dynamic clustering and static clustering. Our major finding is that, when users are dense enough in a network, the proposed method provides the same level of coordination benefit with dynamic clustering to edge users.
\end{abstract}

\section{Introduction}
\subsection{Motivation}


Coordination among base stations (BSs) is a powerful approach for mitigating inter-cell interference in cellular systems \cite{lee:comp2012}. It has been shown that the sum spectral efficiency scales linearly with the signal-to-noise ratio (SNR) when all the BSs are coordinated \cite{somekh:co_multicell_zf}. 
In practice, however, coordination with a large number of BSs may not be feasible due to excessive overhead associated with the coordination, e.g., complexity, channel estimation and channel feedback \cite{lozano:2013:limitcoop}. 
One practical solution for implementing multicell coordination is to form a BS cluster so that a limited number of BSs are coordinated to control intra-cluster interference with a reasonable amount of overhead \cite{lozano:2013:limitcoop, jun:09:network}. Unfortunately, when clustering-based BSs coordination is applied in a static way \cite{Huang:static, ges10},  or equivalently in a random way \cite{6410048} independent to users' conditions, 
the performance is mainly limited by unmanageable out-of-cluster interference \cite{lozano:2013:limitcoop, Huang:static, ges10}.

Dynamic clustering, where a BS cluster is dynamically made by users, was proposed in \cite{Papadogiannis:dclus} to overcome the performance limitation of static clustering. While dynamic clustering improves the signal-to-interference and noise ratio (SINR) of a user significantly, there are several challenges that make it difficult to realize. 
In particular, different users who share the same spectrum can select the same BS at the same time. We refer to this as ``the BS selection conflict problem."
To resolve this, users in a network need to send feedback regarding desirable BS clusters and a global scheduler is required to schedule BSs and users based on this feedback information.
For this reason, performance chiefly depends on the choice of this scheduler. 
The scheduling itself, however, requires a lot of feedback overhead from the users and a substantial amount of coordination of all the BSs and the users, which causes huge system overhead.

To solve the BS selection conflict problem without requiring global coordination, semi-static clustering was proposed in \cite{periodic:frequency, Purmehdi:2013icc}. In this method, multiple predefined BS clusters are used with different time-frequency resources. A set of BS clusters that use the same time-frequency resource is referred as ``a BS cluster pattern.'' For instance, in a square grid BS topology, four BS cluster patterns are required to cover the whole region without the BS selection conflict problem. 
Using the semi-static clustering, users can communicate with their favorable BS set. The main limitation of the existing work \cite{periodic:frequency, Purmehdi:2013icc} is that the BS clustering methods are applicable for a regular BS topology where each BS is located regularly on a grid. In practice, however, the BS topology of cellular networks is irregular \cite{Andrew:10:stocg}. In such case, BS clustering pattern design is challenging. In this work, by leveraging geometry properties, we make progress in designing BS clusters and BS cluster patterns in a semi-static way, which is applicable for any wireless networks whose topology is irregular. 


\subsection{Related Work}
Cooperative strategies for dynamic clustering have been extensively studied in the literature \cite{Papadogiannis:dclus, pa10, wd11, gkb12, 6848065, ry07, ry09, nh10, song:vcellcoop, nrm12, NY:dynamic, junzhang:dynamic, Tanbourgi:2014, 6928420, 6909064, 6879305}. In \cite{Papadogiannis:dclus, pa10, wd11, gkb12}, an algorithm for making a BS cluster in dynamic clustering was proposed. For instance, given the users' conditions such as distances or channel conditions, an algorithm selects BSs to include into a BS cluster so that the joint achievable throughput of the formed BS cluster is maximized.
In \cite{6848065, ry07, ry09}, a resource allocation method to enhance sum spectral efficiency was proposed for dynamic clustering. By jointly allocating resources to each BS cluster depending on their channel qualities, the sum spectral efficiency in a BS cluster can be improved.
In \cite{nh10, song:vcellcoop, nrm12}, a beamforming algorithm 
was proposed in dynamic clustering aiming complexity reduction or sum spectral efficiency improvement.
In this work \cite{Papadogiannis:dclus, pa10, wd11, gkb12, 6848065, ry07, ry09, nh10, song:vcellcoop, nrm12}, a toy network, where only limited number of BSs and users are 
assumed, was considered.
In \cite{NY:dynamic, junzhang:dynamic, Tanbourgi:2014, 6928420}, the performance of dynamic clustering was characterized in a random network where BSs are distributed according to a homogeneous Poisson point process (PPP).
In \cite{6909064}, pair-wise dynamic clustering with rate splitting method was proposed assuming a Poisson network.
In \cite{6879305}, dynamic clustering in two-tier network was proposed and analyzed also in a Poisson network.
The main limitation of the prior work considering BS clustering in a random network \cite{NY:dynamic, junzhang:dynamic, Tanbourgi:2014, 6928420, 6909064, 6879305} is that they did not provide a solution for resolving the BS selection conflict problem, rather they focused on the performance characterization (e.g., the coverage probability) under the assumption that a global scheduler allows such BS clustering to a typical user without the BS selection conflict problem. Global scheduling has higher complexity as overhead and scheduling are required between all possible BS clusters. 
The main difference between this paper and the prior work in \cite{NY:dynamic, junzhang:dynamic, Tanbourgi:2014, 6928420, 6909064, 6879305} is that we focus on providing a systematic solution for the BS selection conflict problem without using a global scheduler, so that the proposed method is able to be implemented in a network with reasonable overhead.

\subsection{Contributions} 
In this paper, we propose a method for designing BS clusters and BS cluster patterns for pair-wise multicell coordination in a cellular downlink network whose BSs' locations are irregular. The goal of the proposed method is to ensure that selected users communicate with their two closest BSs avoiding the BS selection conflict problem. To do this, we 
leverage the concept of 2nd-order Voronoi region. Suppose two BSs are located at ${\bf d}_i \in\mathbb{R}^2$ and ${\bf d}_j \in\mathbb{R}^2$ in a two dimensional space. Then, 
the set of closer points to ${\bf{d}}_{i}$ and ${\bf{d}}_{j}$ than any other points is defined by the 2nd-order Voronoi region as
\begin{align} \label{voronoi_defn}
&\mathcal{V}_2\left({\bf{d}}_i, {\bf{d}}_j \right) \nonumber \\
&= \left\{ {\bf{d}}\in \mathbb{R}^2 \left| 
\begin{array}{lc}
\left\{ \left\|{\bf{d}} - {\bf{d}}_{i} \right\| \le \left\|{\bf{d}} - {\bf{d}}_{k} \right\|  \right\}  \\
\cap \left\{ \left\|{\bf{d}} - {\bf{d}}_{j} \right\|  \le \left\|{\bf{d}} - {\bf{d}}_{k} \right\| \right\},
\\\; \forall k \in \mathbb{Z}^+,\; k\ne i,j  \end{array} \right. \right\},
\end{align}
The core feature of the 2nd-order Voronoi region is that 
it characterizes an area where any point in that area has the minimum distances to $\{{\bf{d}}_{i}, {\bf{d}}_{j}\}$. If a BS pair located at $\{{\bf{d}}_{i}, {\bf{d}}_{j} \}$
forms a BS cluster and serves users in $\mathcal{V}_{2}\left({\bf{d}}_{i}, {\bf{d}}_{j} \right)$, those users are guaranteed to communicate with their two closest BSs. 
Applying this to the whole network,
the first step of the method is tessellating a network plane into the 2nd-order Voronoi regions of every feasible pair of BS locations and making a BS cluster according to 2nd-order Voronoi regions.
Next step is designing BS cluster patterns by allocating time-frequency resources to each BS cluster for preventing the BS selection conflict problem. 
To do this, we adopt the concept of Delaunay triangulation and edge-coloring in graph theory. At the best of authors' knowledge, it is the first work to solve the BS selection conflict problem in irregular BS topologies. By solving this, users are able to obtain BS coordination benefit by communicating with their two closest BSs.
We explain the proposed idea more rigorously in Section II. 

With the proposed method, we characterize the performance of the proposed BS coordination strategy in a fixed irregular network model. Under the premise that multi-user coordinated beamforming (CBF) \cite{chae12} is applied to mitigate intra-cluster interference, we derive analytical expressions for the rate coverage probability and the ergodic spectral efficiency as a function of relevant system parameters: 1) distances between BSs and users, 2) the number of antennas per BS, 3) the number of selected users, 4) signal-to-noise ratio (SNR), and 5) the pathloss exponent. 
Next, considering a random network model where a location of each BS is modeled by using a homogeneous PPP, 
we derive a lower bound on the ergodic spectral efficiency as a function of the relevant system parameters.

To demonstrate efficiency of the proposed method, we use a system level simulation to compare the edge users sum throughput of the proposed method and that of other conventional coordination methods: dynamic clustering and static clustering. 
One main finding is that when users are dense enough, the proposed method achieves the same edge user performance as dynamic clustering in which a global scheduler is assumed.
This gain comes from the fact that the proposed method ensures that selected users communicate with their two closest BSs with high probability as the user density increases.

The remainder of the paper is organized as follows. In Section II, 
the system model including the proposed cluster pattern is explained. 
In Section III, the rate coverage probability is characterized and the ergodic spectral efficiency is obtained in Section IV. In Section V, 
the ergodic spectral efficiency in a random network is analyzed, followed by Section VI which provides comparisons with simulations. Section VII concludes the paper.

\section{System Model} 

In this section, we introduce the system model used in this paper. We first describe the network model, the pair-wise coordination model, and user association assumptions. Then the proposed BS cluster and BS cluster pattern are explained by using a toy example, and extend them to a general network model in the next subsection. Further, we introduce an algorithm for enhancing the proposed method, called as ``edge cutting algorithm."
Signal model and performance metrics are specified in the following subsection.

\subsection{Network Model}
\begin{figure}[t]
\centering
$\begin{array}{cc}
{\resizebox{0.45\columnwidth}{!}
{\includegraphics{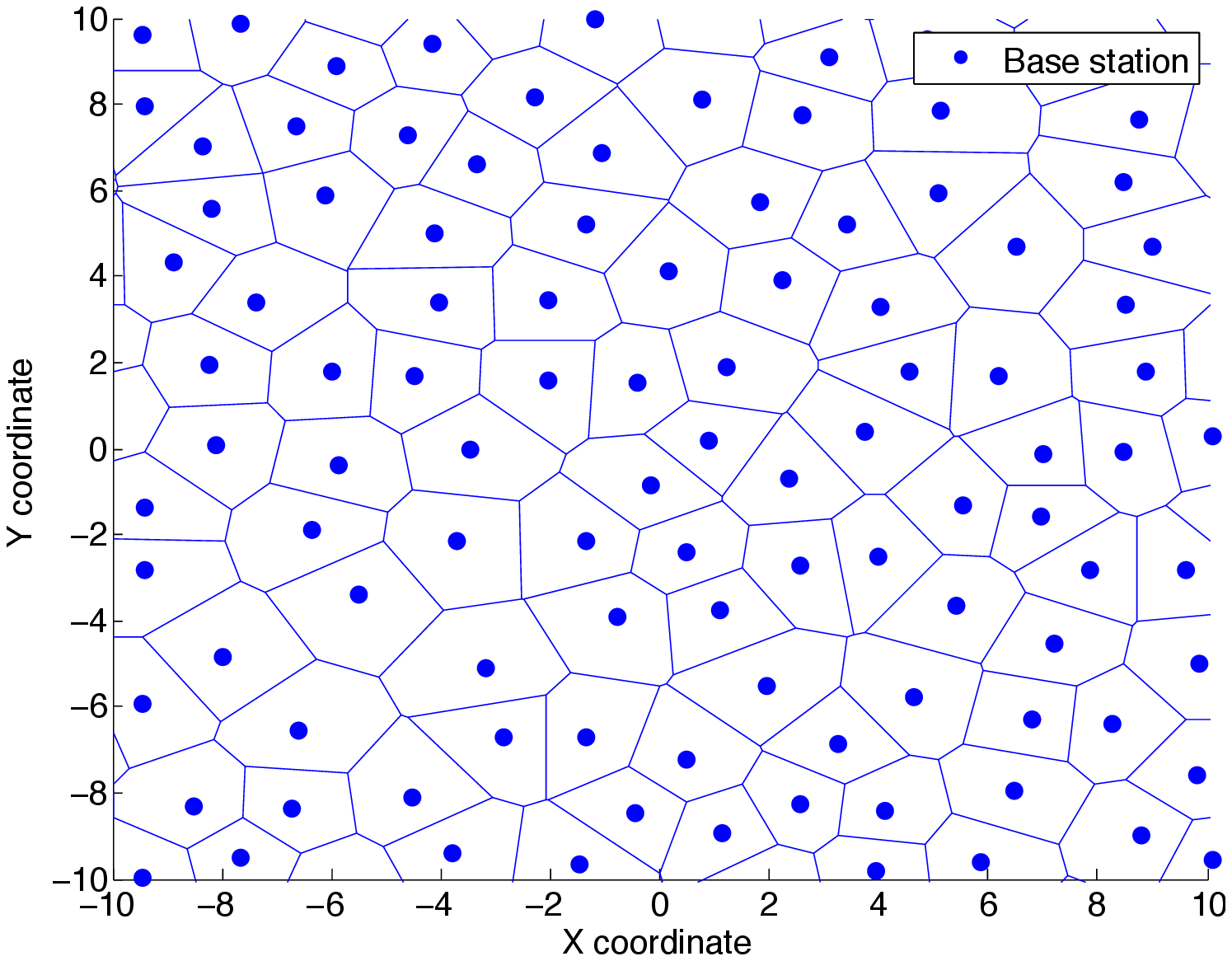}}}  &
{\resizebox{0.45\columnwidth}{!}
{\includegraphics{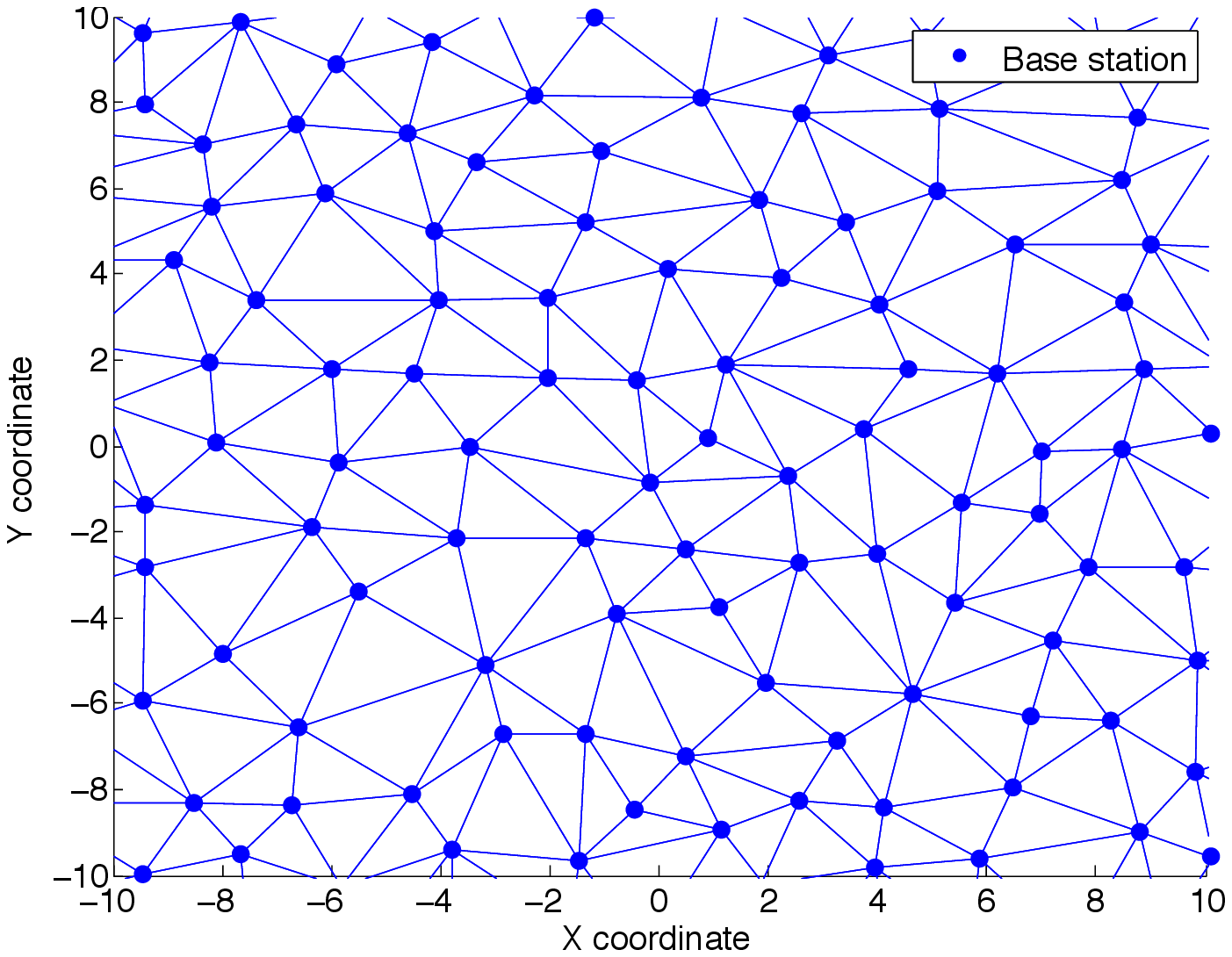}}}  \\ 
\mbox{(a)} &
\mbox{(b)}
\end{array}$
\caption{(a) An example of the considered network model. 
This example is generated by a repulsive point process where a distance between any two BSs is larger than $1.4 \rm km$. 
(b) A graph induced by the Delaunay triangulation corresponding to (a).} 
   \label{2d_net_model}
\end{figure} 

We consider a fixed downlink cellular network where each BS has $N$ antennas.
As illustrated in Fig.~\ref{2d_net_model} (a), BSs are located at irregular points on a two dimensional plane. We denote each location of a BS as ${\bf{d}}_i \in \mathbb{R}^2$ for $i \in \mathbb{Z}^+$, so that $\mathcal{N} = \left\{ \left. {\bf{d}}_{i} \right| i \in \mathbb{Z}^+\right\}$ includes all the locations of the BS on the plane. 
A BS located at ${\bf{d}}_i $ for $i \in \mathbb{Z}^+$ has its own 1st-order Voronoi region ${\mathcal{V}}_1\left({\bf{d}}_i \right)$, and each 1st-order Voronoi region tessellates the network plane into $\left| \mathcal{N} \right|$ regions.


\subsection{Pair-wise Multicell Coordination}
In this paper, we focus on the case where only two adjacent BSs form a cooperative cluster.
The rationale behind focusing on pair-wise BS coordination is that it is highly representative for characterizing coordination gain without causing too much overhead. For example, in \cite{NY:dynamic}, it was shown that a BS cluster including more than two BSs decreased the ergodic spectral efficiency when considering signaling overhead. This is because including more BSs into a BS cluster requires more overhead for estimating channel coefficients of the BSs in the cluster, which decrease the ratio of the transmitted data in a packet. This degrades the spectral efficiency.
For pair-wise multicell coordination, a BS located at ${\bf{d}}_i$ is able to coordinate with a BS located at ${\bf{d}}_j$ if $\lambda\left(\mathcal{V}_2\left({\bf{d}}_i, {\bf{d}}_j \right)  \right) \ne 0$, where $\lambda\left(\cdot\right)$ is the Lebesgue measure in two dimensional space. 

\subsection{User Association}
We assume that $2K$ single antenna users are selected in each 2nd-order Voronoi region irrespective of the size of the region. Among them, $K$ users are associated with the BS composing the corresponding 2nd-order Voronoi region while the other $K$ users are associated with the other BS. For instance, if $2K$ users are in $\mathcal{V}_2\left({\bf{d}}_i, {\bf{d}}_j \right)$, $K$ users in $\mathcal{V}_1\left({\bf{d}}_i\right) \cap \mathcal{V}_2\left({\bf{d}}_i, {\bf{d}}_j \right)$ are associated with the BS located at ${\bf{d}}_i$ and the other $K$ users in $\mathcal{V}_1\left({\bf{d}}_j\right) \cap \mathcal{V}_2\left({\bf{d}}_i, {\bf{d}}_j \right)$ are associated with the BS located at ${\bf{d}}_j$.
This underlying assumption is justified by assuming that the number of users in a network is far larger than the number of deployed BSs, so that there are at least $2K$ users in each 2nd-order Voronoi region satisfying our assumption. More specific user scheduling algorithms are beyond the scope of this paper. 





\subsection{The Proposed Clustering Model - Toy Example}

\begin{figure}[t]
\centerline{\resizebox{0.8\columnwidth}{!}{\includegraphics{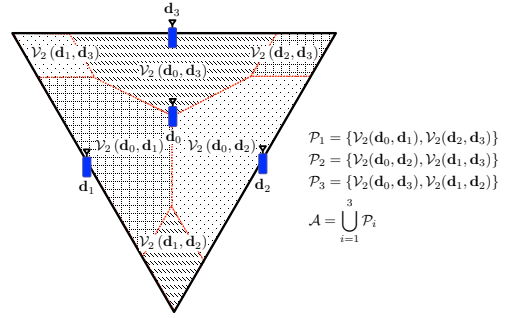}}}    
\caption{A motivational example where all BSs are assumed to be on a finite triangular plane $\mathcal{A}$, each of which is located at ${\bf{d}}_i$ for $i \in \{0,1,2, 3\}$.
The whole plane can be divided by the 2nd-order Voronoi region of $\{{\bf{d}}_i, {\bf{d}}_j \}$ for $i,j \in \{0,1,2,3\}$ and $i \ne j$. 
The areas which has the same shade pattern are assigned to the same cluster pattern, resulting in that three cluster patterns, i.e., $\mathcal{P}_1, \mathcal{P}_2,$ and $\mathcal{P}_3$, are defined. 
}  
   \label{tri_example} 
\end{figure}

In this subsection, we explain the proposed clustering model with a motivational example. To this end, suppose that BS $i$ is located at ${\bf{d}}_i$ in a finite area $ \mathcal{A} \subset{\mathbb{R}^2}$ for $i\in\{0,1,2, 3\}$, as illustrated in Fig.~\ref{tri_example}.
By exploiting the notion of the 2nd-order Voronoi region \eqref{voronoi_defn}, 
it is possible to create $\binom{4}{2} = 6$ different pair-wise cooperative areas $\mathcal{V}_2({\bf{d}}_i, {\bf{d}}_j)$ for $i\neq j$ and $i,j\in\{0,1,2, 3\}$ such that all points within the given region $\mathcal{V}_2({\bf{d}}_i, {\bf{d}}_j)$ have the same two nearest BSs, i.e., BS $i$ and BS $j$. 
In the proposed algorithm, users located in $\mathcal{V}_2({\bf{d}}_i, {\bf{d}}_j)$ are served by a BS pair located at $\{{\bf{d}}_i, {\bf{d}}_j \}$, so that they are guaranteed to communicate with the two closest BSs.
Next, once each of 2nd-order Voronoi region is defined, multiple cluster patterns $\mathcal{P}_{\ell}$ for $\ell = \{1,...,L\}$ are created, each of which contains different 2nd-order Voronoi regions and $\bigcup_{\ell = 1}^{L} \mathcal{P}_{\ell} = \mathcal{A}$. Each cluster pattern uses different time-frequency resources, so that there is no BS conflict between two different cluster patterns.
An example of the BS selection conflict problem occurs when $\mathcal{V}_2\left({\bf{d}}_0, {\bf{d}}_1 \right)\in  \mathcal{P}_1$ and also $\mathcal{V}_2\left({\bf{d}}_0, {\bf{d}}_2 \right)\in  \mathcal{P}_1$. In this case, BS 0 is conflicted, and it cannot serve all the associated users. 
To avoid this, one possible example is to create six cluster patterns such that $\mathcal{P}_1 = \{ \mathcal{V}_2({\bf{d}}_0, {\bf{d}}_1)\} $, $\mathcal{P}_2 = \{ \mathcal{V}_2({\bf{d}}_0, {\bf{d}}_2)\} $, $\mathcal{P}_3 = \{ \mathcal{V}_2({\bf{d}}_0, {\bf{d}}_3)\} $, $\mathcal{P}_4 = \{ \mathcal{V}_2({\bf{d}}_1, {\bf{d}}_2)\} $, $\mathcal{P}_5 = \{ \mathcal{V}_2({\bf{d}}_1, {\bf{d}}_3)\} $, and $\mathcal{P}_6 = \{ \mathcal{V}_2({\bf{d}}_2, {\bf{d}}_3)\} $. 
This example, however, requires too many time-frequency resources, i.e., 6. It is even larger than the case where each BS uses different time-frequency resources without BS coordination, i.e., 4.
A better example is to create three BS cluster patterns such that 
$\mathcal{P}_1=\{\mathcal{V}_2({\bf{d}}_0, {\bf{d}}_1),\mathcal{V}_2({\bf{d}}_2, {\bf{d}}_3)\}$, $\mathcal{P}_2=\{\mathcal{V}_2({\bf{d}}_0, {\bf{d}}_2),\mathcal{V}_2({\bf{d}}_1, {\bf{d}}_3)\}$, and $\mathcal{P}_3=\{\mathcal{V}_2({\bf{d}}_0, {\bf{d}}_3),\mathcal{V}_2({\bf{d}}_1, {\bf{d}}_2)\}$. 
With this example, BS conflicts are prevented and the required number of time-frequency resources is reduced to 3. 
To figure out how much performance gain is achieved by using this cluster pattern, assume that a user is in cell-edge region where the interference chiefly comes from the closest interfering BS. When no BS coordination is applied, the performance of the user might be severely degraded by the dominant interference which has almost same power with the desired signal. When each BS has its own separate resource, the dominant interference can be mitigated but it requires 4 time-frequency resources. By applying the proposed pattern, the dominant interference is mitigated with 3 time-frequency resources, which can lead to the performance improvement.

\subsection{The Proposed Clustering Model - General Network}

\begin{figure}[t]
\centering
{\resizebox{0.8\columnwidth}{!}
{\includegraphics{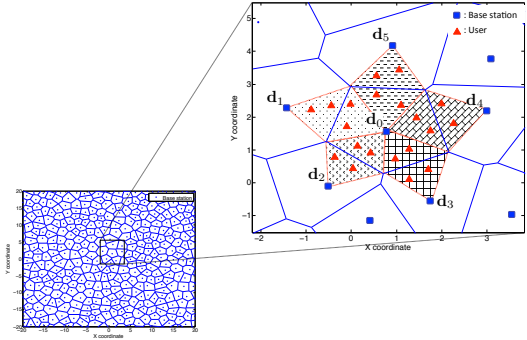}}} 
\caption{An illustration of cooperative areas characterized by the 2nd-order Voronoi region (denoted as dotted red line) focusing on the BS located at ${\bf{d}}_0$. 
Each cooperative area marked by different shade patterns is assigned to a different cluster pattern to avoid BS conflicts. 
} 
   \label{gen_net_example}
\end{figure} 

We extend the proposed clustering concept 
to a general network. At first, we illustrate in Fig.~\ref{gen_net_example} how the proposed BS clusters are formed in a general network. Focusing on a BS located at ${\bf{d}}_0$ and its adjacent BSs at ${\bf{d}}_j$ for $j \in \{1,2,...,5 \}$,
users in $\mathcal{V}_2\left({\bf{d}}_0, {\bf{d}}_j \right)$ for $j \in \{1,2,...,5\}$ are served through a BS pair $\{{\bf{d}}_0, {\bf{d}}_j \}$ by using the proposed clustering method.
After forming BS clusters, we design cluster patterns $\mathcal{P}_{\ell}$ for $\ell \in \{1,2,...,L\}$ where $\mathcal{P}_{\ell} = \{ \left. \mathcal{V}_2\left( {\bf{d}}_i, {\bf{d}}_j \right)\right| i,j \in \mathbb{Z}^+ \}$.
From the intuition of the toy example, 
we summarize the conditions for a desirable cluster pattern.
\begin{enumerate}[1.]
\item $\bigcup_{\ell = 1}^{L} \mathcal{P}_{\ell} = \mathbb{R}^2$: (all users in the network are covered).
\item Any two $\mathcal{V}_2\left({\bf{d}}_i, {\bf{d}}_j \right) \in \mathcal{P}_{\ell}, \mathcal{V}_2\left({\bf{d}}_j, {\bf{d}}_{w} \right)\in \mathcal{P}_{\ell}$ for $\ell \in \{1,...,L\}$, $i \ne j \ne v \ne w$: (to avoid BS conflicts).
\item Small $L$ satisfying 1) and 2): (to minimize the use of unnecessary resources).
\end{enumerate}
In a network whose BSs' locations are irregular, it is challenging to figure out how to map each of the BS clusters into appropriate BS cluster patterns satisfying the above conditions. To solve this, we draw a graph ${G}\left(\mathcal{N}\right)$ whose a vertex is ${\bf{d}}_i \in \mathcal{N}$ and an edge is made by Delaunay triangulation defined in $\mathcal{N}$. Delaunay triangulation for $\mathcal{N}$ is a triangulation of a plane, where each vertex of a triangle is ${\bf{d}}_i \in \mathcal{N}$. One important condition of the Delaunay triangulation is that when drawing the circumcircle of a triangle, there should be no point of $\mathcal{N}$ inside that circumcircle. Note that Delaunay triangulation always can be defined if there are more than three points, i.e., $\left| \mathcal{N} \right| \ge 3$.
The example of $G\left( \mathcal{N} \right)$ is illustrated in Fig.~\ref{2d_net_model} (b), when $\mathcal{N}$ is corresponding to Fig.~\ref{2d_net_model} (a).
A key relationship between 2nd-order Voronoi region and Delaunay triangulation is that
if there is an edge between ${\bf{d}}_i$ and ${\bf{d}}_j$, i.e., $E\left({\bf{d}}_i, {\bf{d}}_j \right)=1$ in $G\left( \mathcal{N} \right)$, the corresponding 2nd-order Voronoi region is non-empty, i.e., $\lambda\left(\mathcal{V}_2\left({\bf{d}}_i, {\bf{d}}_j\right)\right) \ne 0$ \cite{higer_voro}.
For this reason, designing BS cluster patterns, i.e., mapping $\mathcal{V}_2\left({\bf{d}}_i, {\bf{d}}_j\right)$ into a cluster pattern, is equivalent to mapping $E\left({\bf{d}}_i, {\bf{d}}_j \right)$ into a cluster pattern.
When considering the cluster patterns as ``colors,"
the cluster pattern design problem is equivalent to the edge-coloring in graph theory \cite{intro_graph}.
The goal of edge-coloring is to assign colors to the edges of the graph so that any two edges sharing the same vertex have different colors, while minimizing the required number of colors to cover the whole graph. 
By using edge-coloring for $G\left(\mathcal{N} \right)$, we are able to create cluster patterns $\mathcal{P}_{\ell}$, $\ell \in \{ 1,2,...,L\}$ that satisfies the above conditions.

\begin{figure}[t]
\centering
$\begin{array}{cc}
{\resizebox{0.4\columnwidth}{!}
{\includegraphics{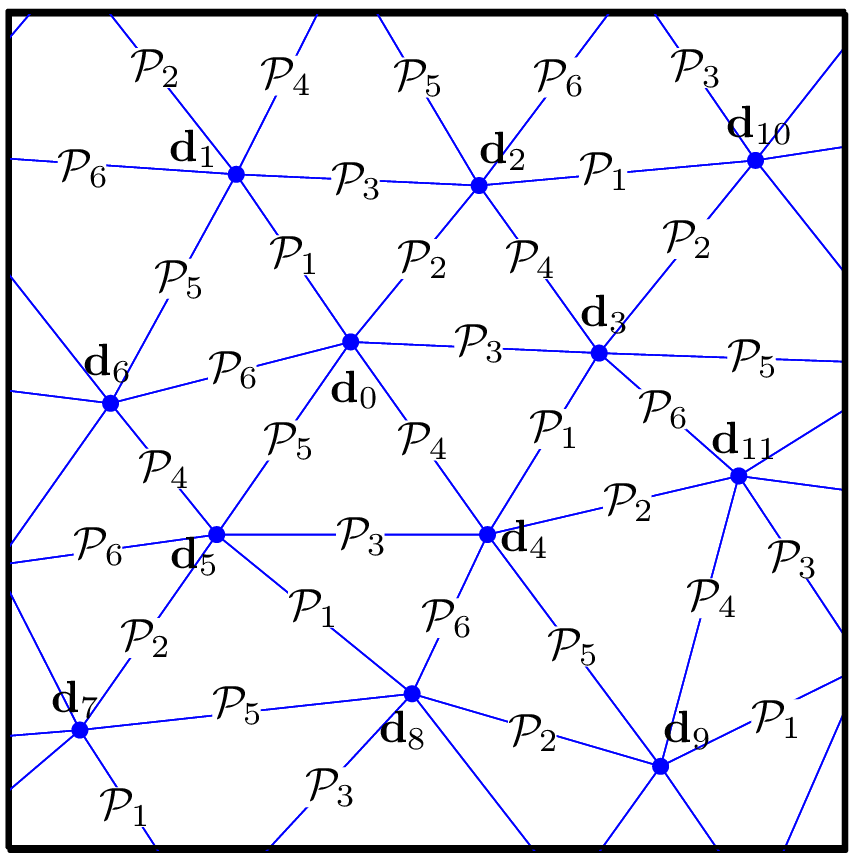}}}  &
{\resizebox{0.4\columnwidth}{!}
{\includegraphics{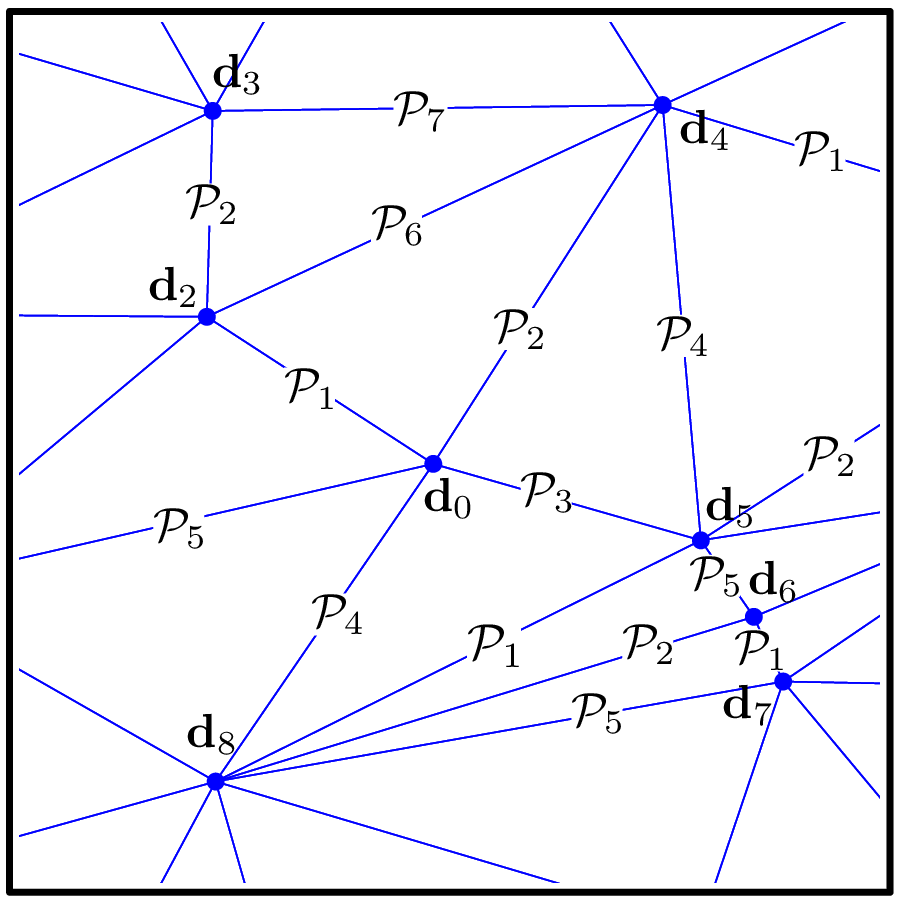}}}  \\  
\mbox{(a)} &  
\mbox{(b)}
\end{array}$
\caption{
Illustration of two different Delaunay triangulations. The membership of the cluster for its cluster pattern is indicated by the edge label.
(a) An example of the Delaunay triangulation of a symmetric network, i.e., the maximum degree is equal to the minimum degree.
(b) An example of the Delaunay triangulation of an asymmetric network, i.e., the maximum degree is strictly larger than the minimum degree. 
} 
   \label{net_example_sym_asym}
\end{figure} 

Now we study the relationship between the cluster pattern design and the edge-coloring specifically. At first, we characterize $L$, which is the required number of different time-frequency resources to avoid the BS selection conflict problem. Thanks to the equivalence between designing cluster patterns and the edge-coloring, $L$ is equal to the minimum required number of colors for the edges. The minimum required number of colors is characterized in the following Theorem.
\begin{theorem}[Vizing's theorem, \cite{vizing}] \label{th_vizing}
A simple planar graph $G$ of maximum degree $\Delta(G)$ has chromatic index $\Delta(G)$ or $\Delta(G) + 1$ in general.
\end{theorem}
\begin{proof}
See the reference \cite{vizing}.
\end{proof}
The chromatic index is known as the required minimum number of colors and the degree means the number of edges connected to the corresponding vertex. In general, it is NP-complete to find the chromatic index given an arbitrary graph, but we can observe that a graph $G$ made by Delaunay triangulation has chromatic index $\Delta(G)$ in many cases.
Now we provide two examples to show how to create the cluster patterns for a given graph. The first example, described in Fig.~\ref{net_example_sym_asym} (a), is a symmetric network where every vertex of $G\left(\mathcal{N}\right)$ has the same degree, 6. This example is fitted to a scenario where each BS is deployed with a guard distance for preventing that two BSs are located very close together.
Specifically, each BS is randomly located within a certain circle, and a center of each circle is located at a grid. 
For this example, $L = 6$. According to the description in Fig.~\ref{net_example_sym_asym} (a), the cluster pattern $\mathcal{P}_{\ell}$ for $\ell \in \{1,2,...,6\}$ can be designed as

\begin{align}
&\mathcal{P}_1 = \{ \mathcal{V}_2 \left({\bf{d}}_0, {\bf{d}}_1 \right), \mathcal{V}_2 \left({\bf{d}}_2, {\bf{d}}_{10} \right), \mathcal{V}_2 \left({\bf{d}}_3, {\bf{d}}_4 \right), \cdots \}, \nonumber \\
&\mathcal{P}_2 = \{ \mathcal{V}_2 \left({\bf{d}}_0, {\bf{d}}_{2} \right), \mathcal{V}_2 \left({\bf{d}}_3, {\bf{d}}_{10} \right), \mathcal{V}_2 \left({\bf{d}}_4, {\bf{d}}_{11} \right), \cdots \}, \nonumber \\
&\mathcal{P}_3 = \{ \mathcal{V}_2 \left({\bf{d}}_0, {\bf{d}}_{3} \right), \mathcal{V}_2 \left({\bf{d}}_1, {\bf{d}}_{2} \right), \mathcal{V}_2 \left({\bf{d}}_4, {\bf{d}}_5 \right), \cdots \}, \nonumber \\
&\mathcal{P}_4 = \{ \mathcal{V}_2 \left({\bf{d}}_0, {\bf{d}}_{4} \right), \mathcal{V}_2 \left({\bf{d}}_2, {\bf{d}}_{3} \right), \mathcal{V}_2 \left({\bf{d}}_5, {\bf{d}}_6 \right), \cdots \}, \nonumber \\
&\mathcal{P}_5 = \{ \mathcal{V}_2 \left({\bf{d}}_0, {\bf{d}}_{5} \right), \mathcal{V}_2 \left({\bf{d}}_1, {\bf{d}}_{6} \right), \mathcal{V}_2 \left({\bf{d}}_7, {\bf{d}}_8 \right), \cdots \}, \nonumber \\
&\mathcal{P}_6 = \{ \mathcal{V}_2 \left({\bf{d}}_0, {\bf{d}}_{6} \right), \mathcal{V}_2 \left({\bf{d}}_4, {\bf{d}}_{8} \right), \mathcal{V}_2 \left({\bf{d}}_3, {\bf{d}}_{11} \right), 
 \cdots \},
\end{align}
by using edge-coloring.
By allocating different time-frequency resources to each cluster pattern $\mathcal{P}_{\ell}$ for $\ell = \{1,2,...,6\}$, no BS conflict occurs and users are guaranteed to communicate with their two closest BSs set.

Now we look an example in Fig.~\ref{net_example_sym_asym} (b), which describes an asymmetric network, i.e., each vertex of $G\left(\mathcal{N}\right)$ can have different numbers of degrees. 
This example is fitted to a case where each BS is deployed with a more random manner than the symmetric network case. 
In this case, $L = 7$ is enough. Similar to the previous example, the cluster pattern $\mathcal{P}_{\ell}$ for $\{1,2,...,7\}$ can be 
\begin{align} \label{ex_clus_pattern_asym}
&\mathcal{P}_1 = \{ \mathcal{V}_2 \left({\bf{d}}_0, {\bf{d}}_2 \right), \mathcal{V}_2 \left({\bf{d}}_5, {\bf{d}}_{8} \right), \mathcal{V}_2 \left({\bf{d}}_6, {\bf{d}}_{7} \right), \cdots \}, \nonumber \\
&\mathcal{P}_2 = \{ \mathcal{V}_2 \left({\bf{d}}_2, {\bf{d}}_{3} \right), \mathcal{V}_2 \left({\bf{d}}_0, {\bf{d}}_{4} \right), \mathcal{V}_2 \left({\bf{d}}_6, {\bf{d}}_{8} \right), \cdots \}, \nonumber \\
&\mathcal{P}_3 = \{ \mathcal{V}_2 \left({\bf{d}}_0, {\bf{d}}_{5} \right), \cdots \}, \nonumber \\
&\mathcal{P}_4 = \{ \mathcal{V}_2 \left({\bf{d}}_4, {\bf{d}}_{5} \right), \mathcal{V}_2 \left({\bf{d}}_0, {\bf{d}}_{8} \right), \cdots \}, \nonumber \\
&\mathcal{P}_5 = \{ \mathcal{V}_2 \left({\bf{d}}_7, {\bf{d}}_{8} \right), \cdots \}, \nonumber \\
&\mathcal{P}_6 = \{ \mathcal{V}_2 \left({\bf{d}}_2, {\bf{d}}_{4} \right),  \cdots \}, \nonumber \\
&\mathcal{P}_7 = \{ \mathcal{V}_2 \left({\bf{d}}_3, {\bf{d}}_{4} \right),  \cdots \}, 
\end{align}
also by using edge-coloring.
No BS conflict occurs as long as a different time-frequency resource is assigned to a different cluster pattern. In this paper, we do not specify an algorithm for solving edge-coloring, rather we assume that edge-coloring is solved perfectly. Some useful algorithms for edge-coloring are found in \cite{gabow_simple, edgecutalgo1}. In \cite{gabow_simple}, an edge-coloring algorithm for simple graphs was proposed and it was proven that the algorithm complexity is $\mathcal{O}(\Delta(G)\left|\mathcal{N}\right| \sqrt{\left| \mathcal{N}\right|\log(\left|\mathcal{N}\right|)})$. In \cite{edgecutalgo1}, an edge-coloring algorithm for planar graphs whose $\Delta(G) \ge 9$ was proposed and the algorithm complexity was proven to be $\mathcal{O}(\left|\mathcal{N}\right| \sqrt{\left| \mathcal{N}\right|\log(\left|\mathcal{N}\right|)})$.

Observe that although the general edge-coloring problem is NP-complete, 
complexity is not a bottleneck since the proposed method (making BS clusters and BS cluster patterns) does not have to be performed frequently. 
Once BS clusters and cluster patterns are made, they are preserved until the geometry of the BSs changes.
Conventionally, the BS geometry would not be changed on the order of months or years.

\subsection{Edge Cutting }
In some BS topologies, particularly in an asymmetric network case, it is possible that unnecessary time-frequency resources are wasted. For instance, in Fig.~\ref{net_example_sym_asym} (b), we found that $L=7$ time-frequency resources are required for the corresponding edge-coloring. Focusing on BS located at ${\bf{d}}_0$ in Fig.~\ref{net_example_sym_asym} (b), however, it only uses $5$ resources therefore $2$ resources are wasted. 
Due to these wasted resources, the network sum throughput even can be degraded.
To relieve this, the edge cutting algorithm is proposed. The key idea of this algorithm is decreasing the maximum degrees $\Delta(G)$ to $\Delta_{\rm EC}(G) < \Delta(G)$ by cutting some edges in the original graph $G(\mathcal{N})$.
Specifically, in $G\left(\mathcal{N}\right)$, if a vertex ${\bf{d}}_i \in \mathcal{N}$ has degree $D\left({\bf{d}}_i \right) > \Delta_{\rm EC}(G)$,
$D\left({\bf{d}}_i \right) - \Delta_{\rm EC}(G)$ edges are selected and cut.
Repeating this to every vertex in a graph, we construct a cut-graph $G^{\rm cut}\left( \mathcal{N}\right)$ 
whose the maximum degree is $\Delta_{\rm EC}(G)$.
Note that the constructed cut-graph $G^{\rm cut}\left( \mathcal{N}\right)$ is a subgraph of $G\left(\mathcal{N} \right)$.
Assuming that an edge $E\left({\bf{d}}_i , {\bf{d}}_j\right)$ is cut in $G^{\rm cut}\left( \mathcal{N}\right)$, users in $\mathcal{V}_2\left({\bf{d}}_i, {\bf{d}}_j \right)$ are 
not served.
By doing this, the required time-frequency resources are $\Delta_{\rm EC}(G)$.
By applying the edge-cutting algorithm, the overall network performance can be improved by saving time-frequency resources. Specifically, the required time-frequency resources $L$ decreases to $\Delta_{\rm EC}(G)$ from $\Delta(G)$ (See Theorem \ref{th_vizing}), where $\Delta_{\rm EC}(G)$ can be controlled as an algorithm parameter.

There are two noticeable points in our edge cutting algorithm. The first one is that selecting the edges to cut has a significant effect to the performance of the algorithm. A desirable way is selecting edges inversely proportional to the area of the corresponding 2nd-order Voronoi region. For instance, assume that we have two candidate edges for cutting, $E\left({\bf{d}}_i, {\bf{d}}_j \right)$ and $E\left({\bf{d}}_i, {\bf{d}}_k\right)$. If $\lambda(\mathcal{V}_2\left({\bf{d}}_i, {\bf{d}}_j \right)) > \lambda(\mathcal{V}_2\left({\bf{d}}_i, {\bf{d}}_k \right))$, we cut $E({\bf{d}}_i, {\bf{d}}_k)$ since its 2nd-order Voronoi area is smaller than the other's area. By doing this, a large cooperative area remains while saving the resources.
The second point is that we can restore edges after finishing the edge cutting. The intuition is that the ideal situation for the proposed method is that every BS has the same degree, as in a symmetric network case Fig.~4 (a). For this reason, we find a vertex whose degree is less than $\Delta_{\rm EC}(G)$ and re-connect edges while satisfying the maximum degree of the graph is $\Delta_{\rm EC}(G)$. The edge restore order is reverse of the edge cut order, so that $E\left({\bf{d}}_i, {\bf{d}}_j\right)$ is restored if $\lambda(\mathcal{V}_2\left({\bf{d}}_i, {\bf{d}}_j \right)) > \lambda(\mathcal{V}_2\left({\bf{d}}_i, {\bf{d}}_k \right))$.
The pseudo code description of the edge-cutting algorithm is described in the following using a MATLAB style notation.

\begin{algorithm}
\caption{Edge-cutting algorithm}
\begin{algorithmic}
\State {\bf{Input}}: $G\left(\mathcal{N}\right)$, $N_{\rm BS} = \left| \mathcal{N}\right| $, $\Delta_{\rm EC}(G)$.
\State {\bf{Initialize}}: Let $S_{\rm Neighbor}({\bf{d}}_i)$, $i \in \mathcal{N}$ be a set of vertices connected via the edge to ${\bf{d}}_i$ in $G(\mathcal{N})$\\
(Area-based edge cutting)
\FOR {$i =1:N_{\rm BS}$} 
\IF {$\left|S_{\rm Neighbor}({\bf{d}}_i) \right| > \Delta_{\rm EC}(G) $}
\State Find ${\bf{d}}_j \in S_{\rm Neighbor}({\bf{d}}_i)$
\State Cut ${\bf{d}}_j$ in order of $(\lambda(\mathcal{V}_2\left({\bf{d}}_i, {\bf{d}}_{j} \right)))$ (ascending order) so that {$\left|S_{\rm Neighbor}({\bf{d}}_i) \right| \le \Delta_{\rm EC}(G) $}
\ENDIF
\ENDFOR
\\ (Area-based edge restore)
\FOR {$i =1:N_{\rm BS}$} 
\IF {$\left|S_{\rm Neighbor}({\bf{d}}_i) \right| < \Delta_{\rm EC}(G) $}
\State Find ${\bf{d}}_j$ cut in the area-based edge cutting
\State Restore ${\bf{d}}_j$ in order of $(\lambda(\mathcal{V}_2\left({\bf{d}}_i, {\bf{d}}_{j} \right)))$ (descending order) satisfying $\max(S_{\rm Neighbor}({\bf{d}}_k)) \le \Delta_{\rm EC}(G),k\in \mathcal{N}$
\ENDIF
\ENDFOR
\State Output: $G^{\rm cut}\left(\mathcal{N} \right)$
\end{algorithmic}
\end{algorithm}
\begin{remark}
Since users in the cut region are not served in the proposed algorithm, there can be a fairness issue. This is able to be resolved by using the proposed method and conventional single cell operation method alternatively. For example, we can use the proposed method during time fraction $\alpha$ and conventional single cell operation during time fraction $(1-\alpha)$. 
Every region is covered during time fraction $(1-\alpha)$. In addition to that, more advanced scheduling algorithm such as proportional fairness can be also incorporated. 
\end{remark}

\begin{remark}
In the proposed edge cutting algorithm, $\Delta_{\rm EC}(G)$ can be determined in various ways. One promising way is to make that all vertices in $G(\mathcal{N})$ have the same degree, so that time-frequency resources are equally allocated to every BS. To do this, $\Delta_{\rm EC}(G)$ should be equal to the minimum degree of $G$, i.e, $\min D({\bf{d}}_i), \; \forall {\bf{d}}_i \in \mathcal{N}$. Other ways to determine $\Delta_{\rm EC}(G)$ can be incorporated depending on BS geometrical features and user density.
For instance, one can formulate an optimization problem as in \cite{6615901}, where $\Delta_{\rm EC}(G)$ is maximized while guaranteeing the difference between $\max D({\bf{d}}_i), \; \forall {\bf{d}}_i \in \mathcal{N}$ and $\min D({\bf{d}}_i), \; \forall {\bf{d}}_i \in \mathcal{N}$ is less than a particular value.
\end{remark}

Finally, we summarize the procedure of the proposed method including the edge cutting algorithm as in the following:
\begin{enumerate}[1.]
\item Tessellate a network plane into 2nd-order Voronoi regions $\mathcal{V}_2({\bf{d}}_i, {\bf{d}}_j)$ for ${\bf{d}}_i, {\bf{d}}_j \in \mathcal{N}$.
\item Make a BS cluster $\{{\bf{d}}_i, {\bf{d}}_j\}$ according to the corresponding $\mathcal{V}_2({\bf{d}}_i, {\bf{d}}_j)$ for ${\bf{d}}_i, {\bf{d}}_j \in \mathcal{N}$.
\item Draw a graph $G(\mathcal{N})$ by Delaunay triangulation.
\item Perform the edge cutting algorithm.
\item By using edge-coloring, make BS cluster pattens $\mathcal{P}_{\ell}$ for $\ell \in \{1,..,L\}$ and assign different time-frequency resource to each BS cluster pattern.
\end{enumerate}
Now we investigate computational complexity of each step. Assuming we have $\left|\mathcal{N}\right|$ points, constructing $\left|\mathcal{N}\right|$ points 2nd-order Voronoi region requires $\mathcal{O}(4\left|\mathcal{N}\right| \log \left|\mathcal{N}\right|)$ \cite{DTL_kNVD}, drawing Delaunay triangulation demands $\mathcal{O}(\left|\mathcal{N}\right| \log \left|\mathcal{N}\right|)$ \cite{Berg:2008:CGA:1370949}, and the edge cutting algorithm needs $\mathcal{O}(\left|\mathcal{N}\right| \Delta(G))$. In general, edge-coloring is NP-complete. 
The complexity is not important part in the proposed method since BS clusters and cluster patterns made by the proposed method are preserved until the geometry of the BSs changes, which occurs on the order of months or years conventionally.

\subsection{Signal Model}

Now we describe the signal model. We focus on a user associated to the BS located at ${\bf{d}}_0$. We denote this user as the tagged user. 
We also assume that the tagged user is located on the origin. This assumption is easily generalized to arbitrary user location by shifting each BS's location to ${\bf{d}}_i - {\bf{u}}$ for $i \in \mathbb{Z}^+$, where ${\bf{u}}$ is a location of the tagged user.
Under the assumption that the tagged user is in $\mathcal{V}_2\left({\bf{d}}_0, {\bf{d}}_1\right)$ and $\mathcal{V}_2\left({\bf{d}}_0, {\bf{d}}_1\right)$ is not cut in the edge cutting algorithm, a BS pair $\{{\bf{d}}_0, {\bf{d}}_{1} \}$ forms a cluster according to the proposed cluster pattern. 
Let's also assume $\mathcal{V}_2\left({\bf{d}}_0, {\bf{d}}_1\right) \in \mathcal{P}_{\ell}$.
In each BS located at ${\bf{d}}_i$, an information symbol vector ${\bf{s}}^{\ell}_{i} \in \mathbb{C}^{K}$, where $ {\bf{s}}^{\ell}_{i} = \left[s^{\ell}_{i,1}, ..., s^{\ell}_{i,K} \right]^T$ for $i \in \mathbb{Z}^+$ is encoded to be transmitted to $K$ users. 
The average power of this transmit symbol vector satisfies $\mathbb{E}\left[ \left\| {\bf{s}}^{\ell}_{i} \right\|^2 \right] \le P$.
When transmitting the symbol vector, a linear beamforming matrix ${\bf{V}}_{i}^{\ell} = \left[{\bf{v}}_{i,1}^{\ell}, ..., {\bf{v}}_{i,K}^{\ell} \right]$, whose ${\bf{v}}_{i,k}^{\ell} \in \mathbb{C}^{N}$ and  $\left\| {\bf{v}}_{i,k}^{\ell}\right\|=1$ for $k = 1,...,K$ is used. 
The observation from the tagged user is given by
\begin{align} \label{2d_sig_model}
y^{\ell}=& \underbrace{\left\| {\bf{d}}_0 \right\|^{-\beta/2}\left({{\bf{h}}}_0^{\ell}\right)^T {\bf{V}}^{\ell}_0{{\bf{s}}}_0^{\ell}}_{ {\rm desired\;signal \;} {\rm +\;inter\;user\;interference} } +  \underbrace{
 \left\| {\bf{d}}_{1} \right\|^{-\beta/2} \left( {{\bf{h}}}^{\ell}_{1}\right)^T {\bf{V}}^{\ell}_{1} {{\bf{s}}}_{1}^{\ell}}_{\rm intra\;cluster\;interference} \nonumber \\   
& +\underbrace{\sum_{{\bf{d}}_j \in \mathcal{N}_{\ell} \backslash \{{\bf{d}}_0, {\bf{d}}_1 \}} \left\| {\bf{d}}_j \right\|^{-\beta/2} \left( {{{\bf{h}}}}^{\ell}_{j} \right)^T {\bf{V}}^{\ell}_{j} {{\bf{s}}}^{\ell}_{j}}_{\rm out\;of\;cluster\;interference}   + n^{\ell} ,
\end{align}
where ${{\bf{h}}_{i}^{\ell}} \in \mathbb{C}^{N}$ is the channel coefficient vector from the BS at ${\bf{d}}_i$ to the tagged user. 
The channel coefficients are assumed to have independent and identically distributed (IID) complex Gaussian entries with zero mean and unit variance, i.e., $\mathcal{CN}\left(0,1\right)$. 
The distance dependent pathloss is with a reference at $1$m and has an exponent of $\beta$ and the additive Gaussian noise is $n^{\ell}$ that follows $\mathcal{CN}\left(0, \sigma^2\right)$. $\mathcal{N}_{\ell} \subseteq \mathcal{N}$ is a set of BSs included in the same cluster pattern with the BS at ${\bf{d}}_0$, i.e., $\mathcal{N}_{\ell} = \{ \left. {\bf{d}}_i \right| \mathcal{V}_2\left({\bf{d}}_i, {\bf{d}}_j \right) \in \mathcal{P}_{\ell}, \forall j\}$. 
The received signal can be separated into three parts: the desired signal and the inter user interference from the associated BS at ${\bf{d}}_0$, the intra-cluster interference from a BS set ${\bf{d}}_1$, and the unmanageable out-of-cluster interference, respectively.

Before transmission of an information symbol, pilot symbols are exploited to learn the downlink channel coefficient vector ${\bf{h}}^{\ell}_0$ and ${\bf{h}}_{1}^{\ell}$.
Then, the user sends the obtained CSI back to the corresponding BSs
via an error-free feedback link. CSI at the transmitter (CSIT) within the same BS cluster is allowed. 
To mitigate intra-cluster interference, CBF is employed by using obtained CSIT.
Assuming that our tagged user's index is $k$, the beamforming matrix ${\bf{V}}^{\ell}_{0}$ and ${\bf{V}}^{\ell}_{1}$ are designed to satisfy the following condition.
\begin{align}
{\rm{maximize}}: &\left| \left({\bf{h}}^{\ell}_0\right)^T {\bf{v}}^{\ell}_{0,k} \right|^2   \\
{\rm{subject\;to}}: & \left| \left( {\bf{h}}^{\ell}_0\right)^T {\bf{v}}^{\ell}_{0, k'} \right|^2 =0\;\;{\rm{for}}\;\; k' \ne k  \\
& \left| \left( {\bf{h}}^{\ell}_{1}\right)^T {\bf{v}}^{\ell}_{1, k''} \right|^2  =0 \;\;{\rm{for}}\;\;
 1\le k'' \le K, 
\end{align}
where the constraints of multi-user CBF are for nullifying the inter-user interference and for removing intra-cluster interference, respectively.
Under the conventional zero forcing (ZF) conditions, 
it is always possible to find such ${\bf{V}}^{\ell}_0$ and ${\bf{V}}_{j}^{\ell}$ if $2K \le N$ with high enough probability if each channel coefficient is IID.

\subsection{Performance Metrics}
If perfect CSIT within the same BS cluster is allowed, then the inter-user interference and intra-cluster interference terms in \eqref{2d_sig_model} are zero. 
Assuming that $\tilde {h}^{\ell}_0$ indicates a modified channel coefficient after multiplying with a beamforming matrix ${\bf{V}}^{\ell}_0$ and $s^{\ell}_{0,0}$ is the information symbol for the tagged user, we have the following modified received signal
\begin{align}
\tilde y^{\ell} =&  \left\|{\bf{d}}_0 \right\|^{-\beta/2} \tilde {h}^{\ell}_0 {s}^{\ell}_{0,0} +\sum_{{\bf{d}}_j  \in \mathcal{N}_{\ell} \backslash \{{\bf{d}}_0, {\bf{d}}_1\}} \left\| {\bf{d}}_j \right\|^{-\beta/2} \left( {{\bf{h}}}^{\ell}_{j}\right)^T {\bf{V}}^{\ell}_{j} {{\bf{s}}}^{\ell}_{j}  \nonumber \\
& + n^{\ell}.
\end{align}
The SINR of the information symbol $s^{\ell}_{0,0}$ is
\begin{align} \label{sinr}
{{\rm SINR}}_{|{\ell}} &=  \frac{ \left\| {\bf{d}}_0 \right\|^{-\beta} \left| \tilde h^{\ell}_0 \right|^2 }{{ \sum_{{\bf{d}}_j \in \mathcal{N}_{\ell} \backslash \{{\bf{d}}_0, {\bf{d}}_1
\}} \left\| {\bf{d}}_j \right\|^{-\beta}\left|\left( {{\bf{h}}}^{\ell}_{j} \right)^T {\bf{V}}^{\ell}_{j} \right|^2}+\sigma^2/P/K  } \nonumber \\
&= \frac{  \left| \tilde h^{\ell}_0 \right|^2 }{I_{{\ell}}+ \left\|{\bf{d}}_0 \right\|^{\beta}/{\rm{SNR}}/K }, 
\end{align}
where $I_{{\ell}} = { \sum_{ {\bf{d}}_j \in \mathcal{N}_{\ell} \backslash \{{\bf{d}}_0, {\bf{d}}_1\}} \left( \left\| {\bf{d}}_j \right\| / \left\| {\bf{d}}_0 \right\| \right)^{-\beta}\left| \left({{\bf{h}}}^{\ell}_{j}\right)^T {\bf{V}}^{\ell}_{j} \right|^2}$, and $P/\sigma^2 = {\rm{SNR}}$. 
Given the system assumptions, the rate coverage probability is given as
\begin{align}
&P^{\ell} \left({\rm{SNR}}, \left\| {\bf{d}}_0 \right\|, \mathcal{D}_{\ell},N, K, \beta, \gamma \right) = \mathbb{P}\left[ \log_2\left(1+  {{\rm SINR}}_{|{\ell}} \right) > \gamma \right], \label{rate_cov_def} 
\end{align}
where $\mathcal{D}_{\ell} = \left\{\left.  {\bf{d}}_j  \right| {\bf{d}}_j \in \mathcal{N}_{\ell} \backslash \{{\bf{d}}_0, {\bf{d}}_1\} \right\}$ and $\gamma$ is the rate threshold. The ergodic spectral efficiency is then
\begin{align} 
&R^{\ell}  \left({\rm{SNR}}, \left\| {\bf{d}}_0 \right\|, \mathcal{D}_{\ell},N, K, L, \beta\right) = \frac{1}{L}\mathbb{E}\left[ \log_2\left(1+  {{\rm SINR}}_{|{\ell}} \right) \right]. \label{rate_def}
\end{align}
The pre-log term $1/L$ is used because $L$ time-frequency resources are used in a network. $L$ is determined as an algorithm parameter in the edge cutting algorithm.

\section{Rate Coverage Analysis}
In this section, we derive the rate coverage probability in a closed form under the interference limited assumption. To this end, we first introduce the following lemma.
Lemma \ref{lem_ch} provides distributions of the desired channel gain and the interference power.


\begin{lemma} \label{lem_ch}
The desired channel gain $\left|\tilde h^{\ell}_0 \right|^2$ follows Chi-squared distribution with $2\left(N - 2K +1\right)$ degrees of freedom.  The out-of-cluster interference $I_{{\ell}}$ follows weighted sum of Chi-squared distribution with $2K$ degrees of freedom. 
\end{lemma}
\begin{IEEEproof}
See \cite{NY:dynamic} and reference therein.
\end{IEEEproof}

By using this Lemma, we derive the rate coverage. The following Theorem is the main result of this section. 

\begin{theorem}
In the interference limited regime, when the BSs form BS clusters by using the proposed method, the rate coverage probability of the tagged user is 
\begin{align} \label{rate_coverage}
&P^{\ell} \left({\rm SNR} \rightarrow \infty, \left\| {\bf{d}}_0 \right\|, \mathcal{D}_{\ell},N, K, \beta, \gamma \right) =  \nonumber \\
&\sum_{m=0}^{N-2K}\frac{\left(2^{\gamma}-1\right)^m}{m!} \left(-1 \right)^m  \left.\frac{\partial^m }{\partial s^m} \prod_{{\bf{d}}_j \in \mathcal{D}_{\ell}}\left( \frac{1}{1+s \left( \frac{\left\|{\bf{d}}_j \right\| }{ \left\| {\bf{d}}_{0} \right\|} \right) ^{-\beta}} \right) ^{K}   \right.,
\end{align}
where $s = 2^{\gamma} - 1$ and $\gamma$ is the rate threshold.
\end{theorem}
\begin{IEEEproof}
As defined in \eqref{rate_cov_def}, the rate coverage probability is rewritten as
\begin{align}
&\mathbb{P}\left[ \log_2\left(1 + {\rm{SINR}}_{|{\ell}}\right)  > \gamma\right] \nonumber \\
& \mathop = ^{(a)}\mathbb{P}\left[ \frac{  \left| \tilde h^{\ell}_0 \right|^2 }{I_{{\ell}}}  >\left(  2^{\gamma} - 1 \right) \right] \nonumber \\
& \mathop =  ^{(b)}\mathbb{E}_{I_{\ell}}\left[\left. \sum_{m=0}^{N-2K}\frac{\left(2^{\gamma} - 1\right)^m}{m!}  \left( I_{\ell}\right) ^{m}  e^{-\left( 2^{\gamma} - 1\right) I_{\ell}}  \right. \right], \label{rate_coverage_pt1}
\end{align}
where (a) comes from that $\left\| {\bf{d}} \right\|^{\beta}/{\rm SNR} \rightarrow 0$ when ${\rm{SNR}} \rightarrow \infty$, and (b) follows Lemma \ref{lem_ch} and the complement cumulative distribution function of a Chi-squared random variable with degrees of freedom $2\left(N-2K + 1\right)$.
Using the derivation of the Laplace transform, i.e., 
$\mathbb{E}\left[ X^m e^{-sX} \right] = \left(-1 \right)^m \frac{d^m \mathcal{L}_X \left( s\right)}{ds^m}$, \eqref{rate_coverage_pt1} is represented as
\begin{align}
&\mathbb{E}_{I_{\ell}}\left[ \sum_{m=0}^{N-2K}\frac{\left(2^{\gamma} - 1\right)^m}{m!}  \left( I_{\ell}\right) ^{m}  e^{-\left(2^{\gamma} - 1 \right)I_{\ell}}  \right] \nonumber \\
&= \sum_{m=0}^{N-2K}\frac{\left(2^{\gamma} - 1\right)^m}{m!} \left(-1 \right)^m  \left.\frac{\partial^m \mathcal{L}_{I_{\ell}} \left( s\right)}{\partial s^m}\right|_{s = 2^{\gamma} - 1},
\end{align}
where $\mathcal{L}_{I_{\ell}}\left(s\right)$ is the Laplace transform of $I_{\ell}$, defined as 
\begin{align}
\mathcal{L}_{I_{\ell}}\left(s\right) = \prod_{{\bf{d}}_j \in \mathcal{D}_{\ell}}\left( \frac{1}{1+s \left( \left\|{\bf{d}}_j \right\| / \left\| {\bf{d}}_{0} \right\| \right) ^{-\beta}} \right) ^{K},
\end{align}
which completes the proof.
\end{IEEEproof}

Computing the rate coverage probability \eqref{rate_coverage} is not easy since computing the derivative of the product of $\left( {1}/{1+s \left( \left\|{\bf{d}}_j \right\| / \left\| {\bf{d}}_{0} \right\| \right) ^{-\beta}} \right) ^{K}$ results in many terms.
To provide a simpler form of the rate coverage probability, the following Corollary is used to approximate the rate coverage probability.

\begin{corollary}
The rate coverage probability \eqref{rate_coverage} can be approximated as in the following form.
\begin{align} \label{rate_coverage_approx}
&P^{\ell} \left({\rm SNR} \rightarrow \infty, \left\| {\bf{d}}_0 \right\|, \mathcal{D}_{\ell},N, K, \beta, \gamma \right) \approx \nonumber \\
& \sum_{m=0}^{N-2K}\frac{\left(2^{\gamma}-1\right)^m}{m!} \left(-1 \right)^m  \left.\frac{\partial^m }{\partial s^m} \left(\frac{ e^{-\sum_{{\bf{d}}_j \in \mathcal{D}_{\ell} \backslash {\bf{d}}_{\rm min}}s \left(\frac{\left\|{\bf{d}}_j \right\|}{\left\| {\bf{d}}_0 \right\|} \right)^{-\beta}  }}{1 + s\left(\frac{ \left\| {\bf{d}}_{\rm min} \right\| }{ \left\| {\bf{d}}_0 \right\| }\right)^{-\beta} }\right)^{K}   \right.,
\end{align}
where $s = 2^{\gamma} - 1$ and ${\bf{d}}_{\min} = \mathop {\arg \min} \limits_{{\bf{d}}_j \in \mathcal{D}_{\ell}} \left\| {\bf{d}}_j\right\|$.
\end{corollary}
\begin{IEEEproof}
Using the Talyor expansion of the exponential function $e^{x} = 1 + \frac{x}{1!} + \frac{x^2}{2!} + \cdots$, when $x$ is small enough, $e^{x} \approx 1 + x$, therefore $e^{-x} \approx 1/\left(1 + x\right)$. Then we have
\begin{align}
&\prod_{{\bf{d}}_j \in \mathcal{D}_{\ell}}\left( \frac{1}{1+s \left( \frac{\left\|{\bf{d}}_j \right\| }{\left\| {\bf{d}}_{0} \right\|} \right) ^{-\beta}} \right) ^{K}\nonumber \\
& = \left( \frac{1}{1 + s\left(\frac{ \left\| {\bf{d}}_{\rm min} \right\| }{ \left\| {\bf{d}}_0 \right\| }\right)^{-\beta} } \right) ^{K} \prod_{{{\bf{d}}_j \in \mathcal{D}_{\ell} \backslash {\bf{d}}_{\rm min}}}\left( \frac{1}{1+s \left( \frac{\left\|{\bf{d}}_j \right\| }{ \left\| {\bf{d}}_{0} \right\|} \right) ^{-\beta}} \right) ^{K} \nonumber \\
&\approx  \left(\frac{1}{1 + s\left(\frac{ \left\| {\bf{d}}_{\rm min} \right\| }{ \left\| {\bf{d}}_0 \right\| }\right)^{-\beta} }\right)^{K}\left( e^{-\sum_{{\bf{d}}_j \in \mathcal{D}_{\ell} \backslash {\bf{d}}_{\rm min}}s \left(\frac{\left\|{\bf{d}}_j \right\|}{\left\| {\bf{d}}_0 \right\|} \right)^{-\beta}  } \right) ^{K},
\end{align}
which completes the proof.
\end{IEEEproof}

Using the Taylor expansion, we approximate a product of $\left( {1}/{1+s \left( \left\|{\bf{d}}_j \right\| / \left\| {\bf{d}}_{0} \right\| \right) ^{-\beta}} \right) ^{K}$ in an exponential form, so that we are able to calculate the derivative of \eqref{rate_coverage_approx} more simply than that of \eqref{rate_coverage}. 

Now we verify the approximated form of the rate coverage probability \eqref{rate_coverage_approx} by comparing to the Monte-Carlo simulation results.
For the verification, we use a network model described in Fig.~\ref{network_snapshot} (b) in Section VI.
As illustrated in Fig.~\ref{rate_coverage_verification}, the analytical approximation tightly matches the simulated rate coverage probability over entire rate threshold region and for different numbers of BS antennas and users settings. 

\begin{figure}[t]
\centerline{\resizebox{0.7\columnwidth}{!}  {\includegraphics{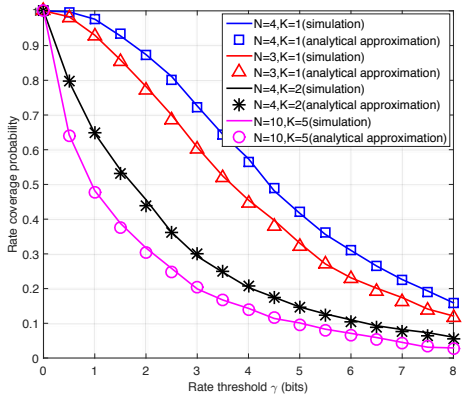}}}  
\caption{The comparison of the rate coverage probability obtained by Monte-Carlo simulations and the analytical approximation \eqref{rate_coverage_approx}. 
It is assumed that $\beta = 4$, $N \in \{3, 4, 10\}$, and $K \in \{1,2, 5\}$.
}    
   \label{rate_coverage_verification}
\end{figure} 

One noticeable point in the rate coverage analysis is that since we do not assume any particular random network, the rate coverage is a function of distances between BSs and the tagged user. Prior work \cite{6967804, 5953530, 5951699} also characterized performance in a fixed network model, where performance was derived as a function of the distances between BSs and a user. 

\section{Ergodic Spectral Efficiency Analysis}
In this section, we derive an exact expression of the ergodic spectral efficiency in an integral form and provide a lower bound on the ergodic spectral efficiency in a closed form. As in the rate coverage analysis, the ergodic spectral efficiency is characterized as a function of distances between BSs and the tagged user.

\subsection{Exact Characterization}
Our exact expression is derived by using 
Lemma \ref{lem_useful}, which yields a general expression of the ergodic spectral efficiency in terms of moment generating functions.

\begin{lemma} \label{lem_useful}
Let $x_1,...,x_N,y_1,...,y_M$ be arbitrary non-negative random variables. Then
\begin{align}
&\mathbb{E}\left[\ln \left( 1+ \frac{\sum_{n=1}^{N}x_n  }{\sum_{m=1}^{M} y_m + 1} \right) \right]  \nonumber \\
&= \int_{0}^{\infty} \frac{\mathcal{M}_y\left(z\right) - \mathcal{M}_{x,y}\left(z\right)}{z}{\rm{exp}}\left(-z\right) {\rm d}z,
\end{align}
where $\mathcal{M}_y\left(z\right) = \mathbb{E}\left[e^{-z\sum_{m=1}^{M}y_m} \right]$ and $\mathcal{M}_{x,y}\left(z\right) = \mathbb{E} \left[ e^{-z\left( \sum_{n=1}^{N} x_n + \sum_{m=1}^{M} y_m \right)} \right]$.
\end{lemma}
\begin{IEEEproof}  
See Lemma 1 in reference \cite{useful}.
\end{IEEEproof}

Leveraging Lemma \ref{lem_useful}, an exact expression of the ergodic spectral efficiency is obtained as in the following Theorem.
\begin{theorem}
When the BSs form BS clusters by using the proposed method, 
the ergodic spectral efficiency of the tagged user is
\begin{align} \label{ergodic_rate_th}
&R^{\ell}  \left({\rm{SNR}}, \left\| {\bf{d}}_0 \right\|, \mathcal{D}_{\ell},N, K, L, \beta\right) \nonumber \\
&=  \frac{\log_2\left(e\right)}{L} \int_{0}^{\infty} \frac{{\rm{exp}}\left(-z \frac{\left\| {\bf{d}}_0 \right\|^{\beta}}{{\rm{SNR}}/K} \right)}{z} \cdot  \nonumber \\
&  \prod_{{\bf{d}}_j \in \mathcal{D}_{\ell}} \left( \frac{1}{1+ z \left(\frac{\left\| {\bf{d}}_j \right\| }{ \left\| {\bf{d}}_0 \right\|} \right)^{-\beta}}\right)^{K}  \left( 1 -  \left( \frac{1}{ 1 + z} \right)^{N-2K+1} \right) {\rm d}z.
\end{align}
\end{theorem}
\begin{IEEEproof}
We start by characterizing the moment generating function of the desired channel gain. Since the desired channel gain is a Chi-squared random variable with $2\left(N-2K+1 \right)$ degrees of freedom from Lemma \ref{lem_ch},
its moment generating function is given by
\begin{align}
\mathcal{M}_S\left(z\right) =& \mathbb{E}\left[e^{-z \left|\tilde h ^{\ell}_0 \right|^2} \right] \nonumber \\
=& \left( \frac{1}{ 1 + z} \right)^{N-2K+1}.
\end{align}
Similarly, the moment generating function of the interference is
\begin{align}
\mathcal{M}_{I_{\ell}}\left(z\right) = \prod_{{\bf{d}}_j \in \mathcal{D}_{\ell}} \left( \frac{1}{1+ z \left(\left\| {\bf{d}}_j \right\| / \left\| {\bf{d}}_0 \right\| \right)^{-\beta}}\right)^{K}.
\end{align}
By applying Lemma \ref{lem_useful}, the Theorem is completed.
\end{IEEEproof}
Given distances from the tagged user to each BS, the ergodic spectral efficiency of the tagged user is obtained by numerically calculating \eqref{ergodic_rate_th}.

\subsection{A Lower Bound Characterization}
We also derive a lower bound on the ergodic spectral efficiency with a closed form which is useful in obtaining insight on how system parameters change the ergodic spectral efficiency. To derive a lower bound, the following Lemma is claimed.
\begin{lemma} \label{lem_lower}
For any non-negative independent random variables $S$ and $I$,
\begin{align}
\mathbb{E}\left[\log_2\left(1 + \frac{S}{I+1} \right) \right] \ge \log_2\left( 1 + \frac{e^{\mathbb{E}[\ln S]}}{\mathbb{E}\left[I+1\right]}\right).
\end{align}
\end{lemma}
\begin{IEEEproof}
\begin{align}
\mathbb{E} \left[\log_2 \left( 1+ \frac{S}{I+1} \right) \right] &  {\ge}  \mathbb{E}_{S} \left[ \mathbb{E} \left[ \left. \log_2\left( 1 + \frac{S}{I+1}\right) \right| S \right] \right]   \nonumber \\
& \mathop {\ge} \limits^{(a)} \mathbb{E}\left[   \log_2 \left( 1 + \frac{e^{\ln S}}{\mathbb{E}\left[I+1\right]}\right) \right]  \nonumber \\
& \mathop {\ge} \limits^{(b)} \log_2\left( 1 + \frac{e^{\mathbb{E}[\ln S]}}{\mathbb{E}\left[I+1\right]}\right)
\end{align}
where (a) follows Jensen's inequality and $\log_2 \left( 1 + \frac{1}{x}\right)$ is a convex function of non-negative variable $x$, and (b) also follows Jensen's inequality and $\log_2 \left( 1 + e^{y}\right)$ is a convex function of non-negative variable $y$.
\end{IEEEproof}
Applying Lemma \ref{lem_lower}, a lower bound of the ergodic spectral efficiency is obtained in the following Theorem.
\begin{theorem} \label{th_low}
When the BSs form a BS coordination set according to the proposed cluster pattern, 
the ergodic spectral efficiency of the tagged user is lower bounded as
\begin{align}
&R^{\ell}  \left({\rm{SNR}}, \left\| {\bf{d}}_0 \right\|, \mathcal{D}_{\ell},N, K, L, \beta\right)  \nonumber \\
&\ge \frac{1}{L}\log_2\left(1 + \frac{{\rm{exp}}\left( \psi\left(N-2K+1 \right)  \right)  }{ K \sum_{{\bf{d}}_j \in \mathcal{D}_{\ell}}\left(\frac{\left\| {\bf{d}}_j \right\| }{ \left\| {\bf{d}}_0 \right\|} \right)^{-\beta} + \left\| {\bf{d}}_0 \right\|^{\beta} / {\rm{SNR}}/K}  \right), \label{lower_bound_rate}
\end{align}
where $\psi\left(z\right)$ is the digamma function defined as $\Gamma'\left(z\right) / \Gamma\left(z\right)$, where $\Gamma\left(z\right)$ is the gamma function defined as
\begin{align}
\Gamma\left(z\right) = \int_{0}^{\infty} x^{z-1}e^{-z} {\rm d} z.
\end{align}
\end{theorem}
\begin{IEEEproof}
First, we characterize $\mathbb{E}\left[ \ln \left|\tilde h ^{\ell}_0 \right|^2\right]$.
Since $\left|\tilde h ^{\ell}_0 \right|^2 \sim \chi^2\left(2\left(N-2K+1\right) \right)$, 
\begin{align}
\mathbb{E}\left[ \ln \left|\tilde h ^{\ell}_0 \right|^2\right] =  \psi\left(N-2K+1 \right),
\end{align}
where $\psi\left(z\right)$ is the digamma function defined as $\Gamma'\left(z\right) / \Gamma\left(z\right)$. Next, we compute $\mathbb{E}\left[I_{{\ell}} \right]$. Since the interference is a sum of independent Chi-squared random variables with different weights corresponding to the ratio of the distances, we have
\begin{align}
\mathbb{E}\left[I_{{\ell}} \right] = K \sum_{{\bf{d}}_j \in \mathcal{D}_{\ell} }\left(\frac{\left\| {\bf{d}}_j \right\|}{ \left\| {\bf{d}}_0 \right\|} \right)^{-\beta}.
\end{align}
By applying Lemma \eqref{lem_lower}, we complete the proof. 
\end{IEEEproof}

\begin{figure}[t]
\centerline{\resizebox{0.7\columnwidth}{!}  {\includegraphics{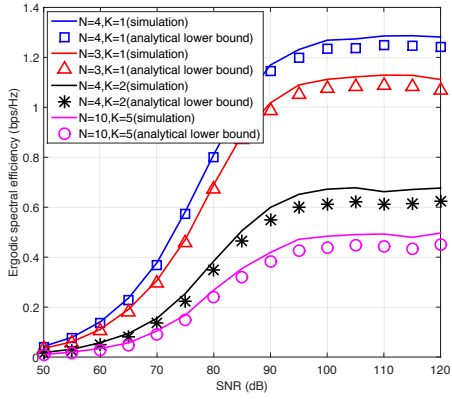}}}  
\caption{The comparison of the ergodic spectral efficiency obtained by Monte-Carlo simulations and the analytical lower bound \eqref{lower_bound_rate}. 
It is assumed that $\beta = 4$, $N \in \{3, 4, 10\}$, and $K \in \{1,2, 5\}$.
}    
   \label{ergodic_verification}
\end{figure} 

Now we verify the derived lower bound by comparing with the ergodic spectral efficiency obtained by using Monte-Carlo simulations. We use a network model described in Fig.~\ref{network_snapshot} (b) in Section VI for the verification. It is assumed that $L = \Delta_{\rm EC}(G) = 4$. The comparison result is in Fig.~\ref{ergodic_verification}. In Fig.~\ref{ergodic_verification}, we observe two separate regime depending on SNR, which are referred as the DoF regime and the saturation regime in \cite{lozano:2013:limitcoop}. In both regimes, the lower bound is reasonably tight for different numbers of BS antennas and users settings.

\section{Ergodic Spectral Efficiency in Random Networks}


So far, we characterize the performance when the BSs are coordinated with the proposed cluster patterns in a fixed network model. In this section, we analyze the ergodic spectral efficiency when applying the proposed method in a random network.
We skip characterizing the rate coverage probability of the proposed method in a random network since it is equivalent to a special case of (26) and (27) in \cite{NY:dynamic}. 
We consider a downlink cellular network where BSs equipped with $N$ antennas are distributed according to a homogeneous PPP, $\Phi = \left\{ \left. {\bf{d}}_{i} \right| i \in \mathbb{N} \right\}$ with density $\lambda$ on the plane $\mathbb{R}^2$. 
For easy of notation, we start the BS location index with $1$ so that the $k$-th closest BS from the origin is denoted as ${\bf{d}}_{k}$.
Single antenna users are also distributed as a homogeneous PPP, $\Phi_{\rm{U}} = \left\{ \left. {\bf{u}}_{i} \right| i \in \mathbb{N} \right\}$ with density $\lambda_{{\rm{U}}}$, which is independent with $\Phi$. We assume that $\lambda_{\rm U} \gg \lambda$ so that there exist at least $K$ users in each region with high enough probability. 
Under this assumption, $K$ users are selected in each region.
In this network model, the proposed method is applied in the same way with the case of a fixed network model, i.e., make BS clusters by using 2nd-order Voronoi region, draw a graph $G\left(\Phi \right)$ by Delaunay triangulation, and design BS cluster patterns by using the corresponding edge-coloring.
Under the assumption that the proposed method is used,
we derive a lower bound on the ergodic spectral efficiency in a closed form.
One should note that the randomness of the network model in this section is not a core part of the proposed method. Rather, it is a tool for analyzing the ergodic spectral efficiency performance.

Without loss of generality, we focus on the typical user located at the origin, and assume that it is associated with the BS located at ${\bf{d}}_1$. Under the premise that the typical user is in $\mathcal{V}_2\left({\bf{d}}_1, {\bf{d}}_2\right) \in \mathcal{P}_{\ell}$, the typical user is served by a BS cluster $\{ {{\bf{{d}}}_1},{{ \bf{{d}}}_2} \}$.
Similarly to the fixed network case, the typical user communicates with a BS pair $\{ {{\bf{{d}}}_1},{{ \bf{{d}}}_2} \}$ by applying CBF, so that the intra-cluster interference from the BS at ${\bf{d}}_2$ to the typical user is removed.
Then, equivalently with \eqref{sinr}, the SINR of the typical user is given by
\begin{align} \label{sir_random}
{\rm SINR}_{\left| \ell \right.} = \frac{\left\| {\bf{d}}_1 \right \|^{-\beta} \left| \tilde h^{\ell}_1 \right|^2}{\sum_{{\bf{d}}_j \in \Phi \backslash \mathcal B\left(0, \left\| {\bf{d}}_2 \right\| \right)}{ \left\| {\bf{d}}_j \right\|^{-\beta}\left| \left({{\bf{h}}}^{\ell}_{j}\right)^T {\bf{V}}^{\ell}_{j} \right|^2} + 1/ {\rm{SNR}}/K },
\end{align}
where $\tilde h^{\ell}_1$ is the modified channel coefficient from the BS at ${\bf{d}}_1$ to the typical user after applying CBF, ${\bf{h}}_j^{\ell}$ is the channel coefficient vector from the BS at ${\bf{d}}_j$ to the typical user, and ${\bf{V}}^{\ell}_{j}$ is a beamforming matrix of the BS at ${\bf{d}}_j$, respectively. 
The entries of the channel coefficient vector follows $\mathcal{CN}\left(0,1\right)$. From \eqref{sir_random}, the ergodic spectral efficiency is defined as
\begin{align} \label{rate_random}
&R^{\ell} \left({\rm{SNR}},N, K, \beta\right) = \frac{1}{L}\mathbb{E}\left[\log_2\left(1+  {{\rm SINR}}_{|{\ell}} \right) \right],
\end{align}
where the pre-log term $1/L$ is used because $1/L$ time-frequency resources are exploited to serve the typical user. 

\subsection{Lower Bound on the Ergodic Spectral Efficiency}
Now we derive a lower bound on the ergodic spectral efficiency. 
To do this, we need the following Lemma. 
Lemma \ref{lem_joint_pdf} provides probability density function (PDF) of distances of the closest BS and $k$-th closest BS to the origin.

\begin{lemma} \label{lem_joint_pdf}
The joint probability density function (PDF) of $\left\| {\bf{d}}_1 \right\|$ and $\left\|{\bf{d}}_k \right\|$ is 
\begin{align}
&f_{\left\| {\bf{d}}_1 \right\|, \left\| {\bf{d}}_k \right\|}\left(r_1, r_k \right) \nonumber \\
&= \left\{ \begin{array}{cc} \frac{4\left(\lambda \pi \right)^k}{\left(k-2 \right)!} r_1 r_k \left(r_k^2 - r_1^2 \right)^{k-2} e^{-\lambda \pi r_k^2} & {\rm if}\; r_1 \le r_k \\ 0&{\rm otherwise.}  \end{array} \right.
\end{align}
\end{lemma}
\begin{proof}
See Appendix C in \cite{NY:dynamic}.
\end{proof}

Leveraging Lemma \ref{lem_joint_pdf}, Theorem \ref{theo_ppp_ergodic} is claimed. 
For analytical tractability, we assume the interference limited regime.

\begin{theorem} \label{theo_ppp_ergodic}
When BSs' locations are distributed by a homogeneous PPP and the proposed method is applied,
the ergodic spectral efficiency of the typical user is lower bounded as
\begin{align} \label{theo_low_random}
&R^{\ell} \left( N, K, \beta\right) \nonumber \\
&\ge \frac{1}{L} \log_2 \left(1 + \frac{\left(\beta^2 - 4 \right)}{{8K}} {\exp\left( \psi\left(N-2K+1 \right)\right)} \right),
\end{align}
where $\psi \left( z \right)$ is the digamma function, defined as $\Gamma'\left(z\right) / \Gamma\left(z\right)$, where $\Gamma\left(z\right)$ is the gamma function defined as
\begin{align}
\Gamma\left(z\right) = \int_{0}^{\infty} x^{z-1}e^{-z} {\rm d} z.
\end{align}
\end{theorem}
\begin{proof}
The ergodic spectral efficiency \eqref{rate_random} is written by
\begin{align} \label{theorem4_frame}
&\mathbb{E}\left[ \frac{1}{L} \log_2\left(1+  {{\rm SINR}}_{|{\ell}} \right) \right] \nonumber \\ 
&= \frac{1}{L} \mathbb{E} \left[ \log_2\left(1+  {{\rm SINR}}_{|{\ell}} \right) \right] \nonumber \\
&\mathop {\ge} \limits^{(a)} \frac{1}{L} \log_2\left(1 + \frac{e^{\mathbb{E} \left[ \ln\left( \left| \tilde h^{\ell}_1 \right|^2 \right)  \right]}}{A} \right),
\end{align}
where $A = \mathbb{E} \left[\left\| {\bf{d}}_1 \right \|^{\beta}   \sum_{{\bf{d}}_i \in \Phi \backslash \mathcal B\left(0, \left\| {\bf{d}}_2 \right\| \right)}{ \left\| {\bf{d}}_i \right\|^{-\beta}\left| \left({{\bf{h}}}^{\ell}_{i}\right)^T {\bf{V}}^{\ell}_{i} \right|^2}\right]$, (a) follows Lemma \ref{lem_lower}, and $L = \Delta_{\rm EC}(G)$. Note that $\Delta_{\rm EC}(G)$ is determined as an algorithm parameter of the edge-cutting algorithm.
Now, we obtain the expectation in the nominator and denominator of \eqref{theorem4_frame}.
Since the interference is mitigated by CBF, as in the case of the deterministic square grid network model, $\left|\tilde h ^{\ell}_1 \right|^2$ follows $\chi^2\left(2\left(N-2K+1\right) \right)$, which leads to 
\begin{align} \label{theorem4_X}
\mathbb{E} \left[ \ln\left( \left| \tilde h^{\ell}_1 \right|^2 \right)  \right] = \psi\left(N-2K+1 \right),
\end{align}
where $\psi\left(\cdot\right)$ is the digamma function. 
Next, we calculate the expectation of the interference part. We write the expectation of the denominator of the SIR as
\begin{align} 
&{\mathbb{E} \left[\left\| {\bf{d}}_1 \right \|^{\beta}   \sum_{{\bf{d}}_i \in \Phi \backslash \mathcal B\left(0, \left\| {\bf{d}}_2 \right\| \right)}{ \left\| {\bf{d}}_i \right\|^{-\beta}\left|\left({{\bf{h}}}^{\ell}_{i}\right)^T {\bf{V}}^{\ell}_{i} \right|^2}\right]} \nonumber \\
&\mathop = \limits^{(a)}  K {\mathbb{E} \left[\left\| {\bf{d}}_1 \right \|^{\beta}   \sum_{{\bf{d}}_i \in \Phi \backslash \mathcal B\left(0, \left\| {\bf{d}}_2 \right\| \right)}{ \left\| {\bf{d}}_i \right\|^{-\beta}}\right]} \nonumber \\
\mathop = \limits^{} & K \mathbb{E}_{r_1, r_2} \left[ {\mathbb{E}_{\Phi \backslash \mathcal{B}\left(0, r_2 \right)} \left[\left. r_1^{\beta} \sum_{{\bf{d}}_i \in \Phi \backslash \mathcal B\left(0, r_2 \right)}{ \left\| {\bf{d}}_i \right\|^{-\beta}}
\right| \begin{array}{c}{\left\| {\bf{d}}_1 \right\| = r_1}, \\{ \left\| {\bf{d}}_2 \right\| = r_2}\end{array}\right]} 
\right] \nonumber \\
\mathop = \limits^{(b)} & K \mathbb{E}_{r_1, r_2} \left[ r_1^{\beta} 2 \pi \lambda \int_{r_2}^{\infty} r^{1-\beta}  {\rm d} r \right] \nonumber \\
\mathop = \limits^{(c)} & \frac{ 2 K\pi \lambda}{\beta-2} \int_{r_2 = 0}^{\infty} \int_{r_1 = 0} ^{r_2}  4\left(\lambda \pi \right)^2   e^{-\lambda \pi r^2_2} r _1^{\beta+1} r_2^{3-\beta}{\rm d} r_1 {\rm d} r_2 \nonumber \\
=& \frac{8K}{\beta^2 - 4} \label{theorem4_Y}
\end{align}
where (a) follows $\left|\left({{\bf{h}}}^{\ell}_{i}\right)^T {\bf{V}}^{\ell}_{i} \right|^2 \sim \chi^2\left(2K \right)$ and it is independent to $\Phi$, (b) follows the Campbell's theorem, and (c) follows that the joint PDF in Lemma \ref{lem_joint_pdf}. 
Plugging \eqref{theorem4_X} and \eqref{theorem4_Y} into \eqref{theorem4_frame} the proof is completed.
\end{proof}

One should note that Theorem \ref{theo_ppp_ergodic} considers the worst case for the out-of-cluster interference. This is because, in an irregular and asymmetric network model, e.g., a homogeneous PPP model, some BSs do not use the specific time-frequency resource exploited to serve the typical user. Then, the interference from those BSs should not be counted in the out-of-cluster interference term. Analyzing this, however, is highly difficult because it depends on the correlation of the distributed points. It is not tractable in the known PPP framework. For this reason, we rather consider the worst case scenario, which assumes that the out of cluster interference comes from all the BSs in a network.

\section{Performance Comparisons}
\subsection{Simulation Setup}
We provide performance comparison results to demonstrate the superiority of the proposed method via system level simulations. As a simulation setup, we consider a $1400{\rm m} \times 1400{\rm m}$ square network plane, that consists of $49$ ($= 7 \times 7$) square grids whose a side length is $200 {\rm m}$. 
Within each square grid, one BS is randomly located. A range of a BS location is restricted by the perturbation square whose a side length is $p$, $0 \le p \le 200$, and a center is same with the square grid. Given $p$, a BS can randomly located inside $p {\rm m}\times p{\rm m}$ perturbation square. For example, if the perturbation square has $p = 0 $, 
each BS is regularly located on a center of each square grid. This is equivalent to a regular network. In other extreme case, if the perturbation parameter has $p = 200 $, each BS can be located in any place in each square grid, which causes maximal randomness in a network. Fig.~\ref{perturb_square} describes the perturbation square and the corresponding BS location. This setting mimics an actual BS deployment scenario, where a BS is built to satisfy a certain guard region to other BSs, preventing that any two BSs are closely located each other.

For the simulation, we consider two cases, $p = 100$ and $p = 200$. 
A case of $p=200$ represents more irregularity of BSs locations than a case of $p=100$. In these networks, $K_{\rm perBS} \times 49$ users are uniformly dropped. 
Network snapshots for these two cases are illustrated in Fig.~\ref{network_snapshot}.
For these snapshots, $K_{\rm perBS}=10$. 
Clearly, we observe that BSs are more irregularly located when $p=200$.
In the performance comparison, we only consider users inside $1000 {\rm m} \times 1000 {\rm m}$ for wrap-around. 
The number of BS antennas and selected users are assumed to be $N = 3$ and $K = 1$. The path-loss exponent is $4$.

\begin{figure}[t]
\centerline{\resizebox{0.5\columnwidth}{!}{\includegraphics{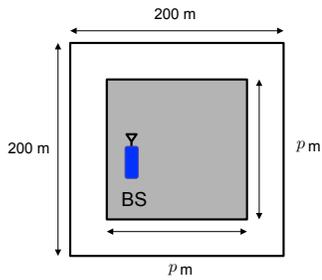}}}    
\caption{A description of the perturbation square whose a side length is $p$. In the described case, a BS is randomly located inisde a region marked by grey shade. }  
   \label{perturb_square} 
\end{figure}

Applying the proposed method in a given network, it is set that $L = \Delta_{\rm EC}(G) = 4$ in both cases of $p=100$ and $p=200$ for the edge cutting algorithm. Recalling the edge cutting algorithm, we cut edges inversely proportional to the area of the corresponding 2nd-order Voronoi regions. However, measuring the area of 2nd-order Voronoi region is a challenging task. For this reason, instead of directly measuring each area, we uniformly drop dummy users and count the number of users associated with the certain two closest BSs pair. Denoting that the number of users whose the two closest BSs are $\{{\bf{d}}_i, {\bf{d}}_j\}$ as $N({\bf{d}}_i, {\bf{d}}_j )$, we approximate $\lambda(\mathcal{V}_2({\bf{d}}_i, {\bf{d}}_j)) \approx N({\bf{d}}_i, {\bf{d}}_j )$. In the simulation, we drop $5000$ dummy users for the approximation.

\begin{figure}[t] 
\centering 
$\begin{array}{cc}  
{\resizebox{0.4\columnwidth}{!}
{\includegraphics{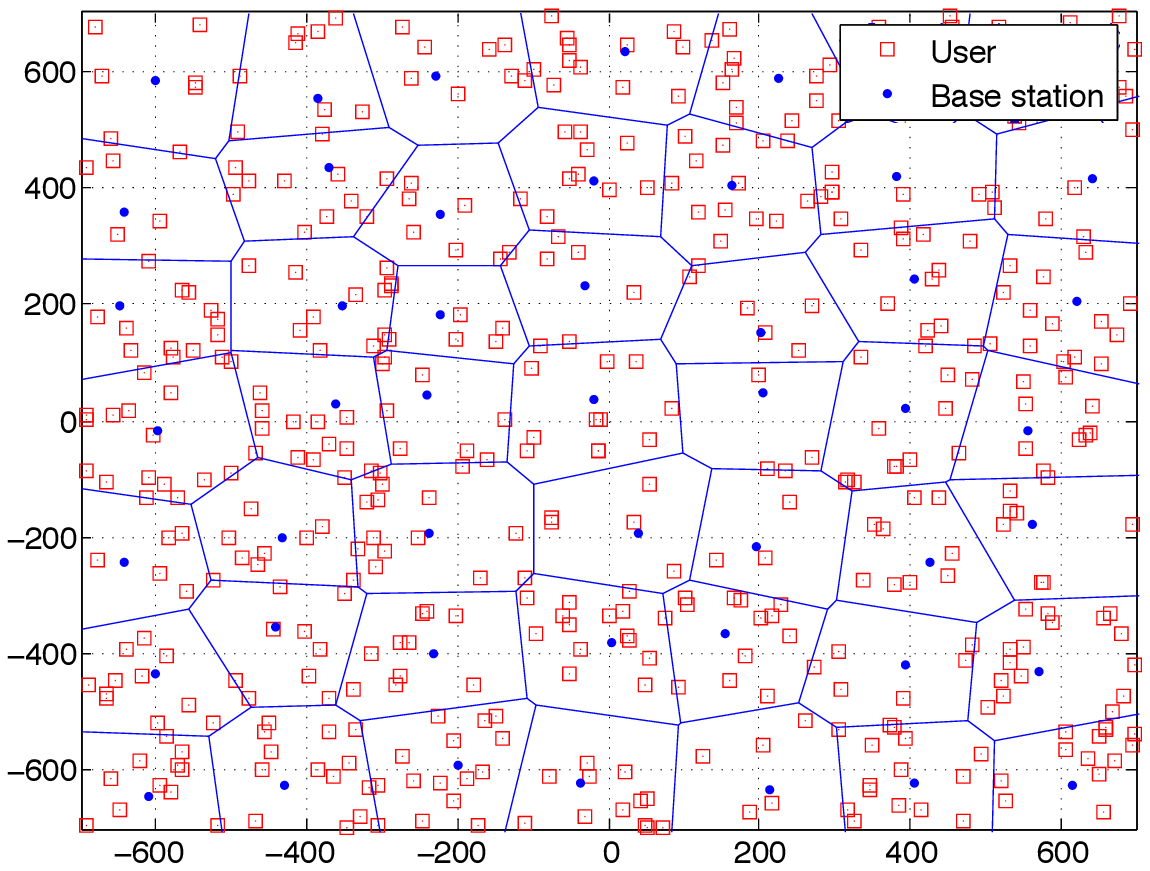}}}  &  
{\resizebox{0.4\columnwidth}{!}
{\includegraphics{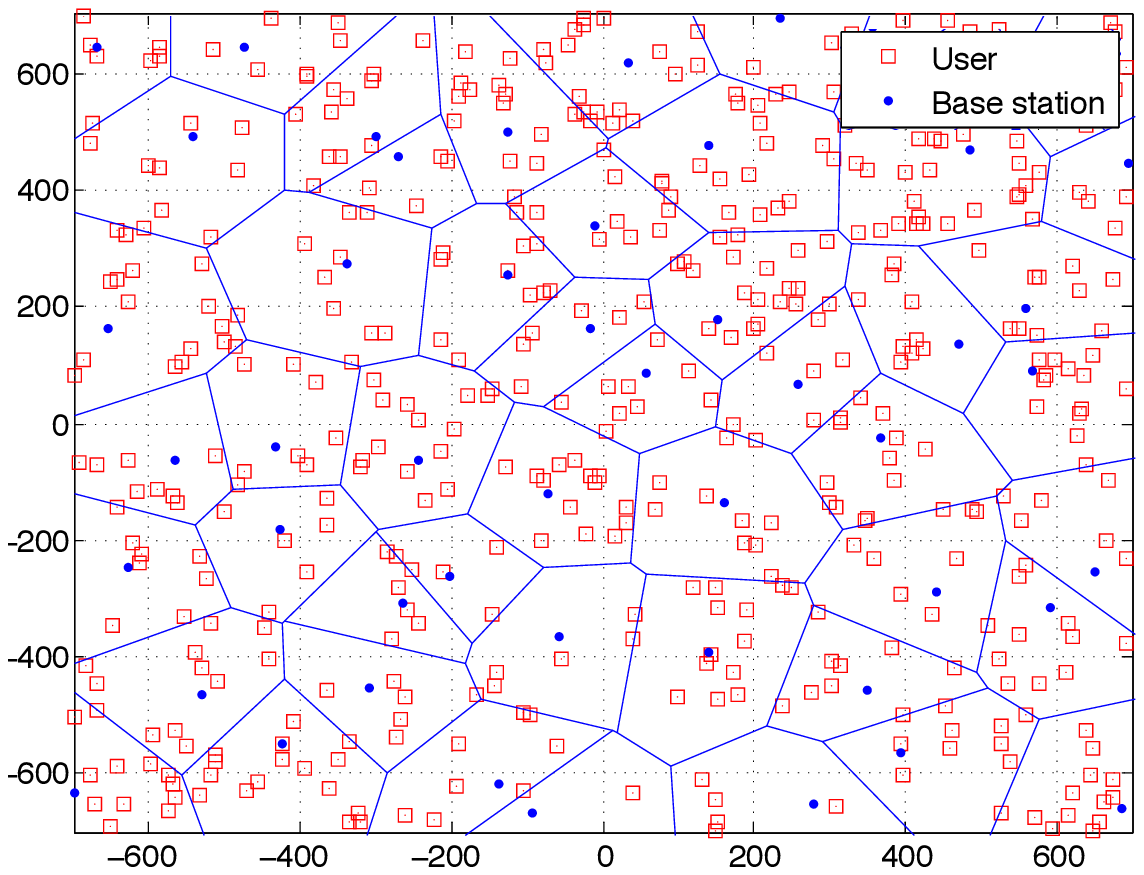}}}  \\  
\mbox{(a)} &
\mbox{(b)}
\end{array}$
\caption{Network snapshots for $p=100$ in (a) and $p=200$ in (b). It is assumed that $K_{\rm perBS} = 10$. Each network plane is tessellated by 1st order Voronoi region of each BS. }
\label{network_snapshot}  
\end{figure} 

For the performance comparison, we consider two conventional coordination methods: dynamic clustering and static clustering. We explain as follows:
\subsubsection{Dynamic Clustering} 
In this method, selected users communicate with their two closest BSs.
For comparison, we propose a simple user scheduling method by modifying FlashLinQ \cite{flashlinq} for dynamic clustering.
Assuming a global scheduler, it assigns random orders to dropped users and the orders can be changed over time for user fairness. Then, according to the assigned order, a user finds the two closest BSs and create a BS cluster. If two users select the same BS, i.e., the BS selection conflict problem occurs, a user whose an order is late gives up communication in this particular time slot and wait for the next time slot. Since two users are able to communicate with one BS cluster ($K=1$ for one BS and a BS cluster has two BSs), once a BS cluster is formed by a user, a global scheduler immediately finds the other user to communicate with the corresponding BS cluster. This process continues until the number of made BS clusters is maximal.
Once all the BS clusters are made, selected users communicate with their BS clusters by using CBF, same with the proposed method.
If there are remaining BSs that could not make a BS cluster, a global scheduler finds a user to communicate with them. No BS coordination is used for these users.

\subsubsection{Static Clustering}
In this method, BS clusters are predetermined irrespective of users' conditions. A BS randomly selects one BS among its neighbor BSs, and makes a BS cluster with the selected one. Then, a user is independently selected by each BS and communicates with the predetermined BS cluster. Similarly to the proposed method and dynamic clustering, the intra-cluster interference is mitigated by using CBF. In the static clustering, there is non-zero probability that a user cannot communicate with its two closest BSs since BS clusters are made independently to users' locations.


\subsection{Simulation Results}
First, we compare the edge users sum throughput depending on $K_{\rm perBS}$. 
An edge user is defined as a user whose a distance ratio between the closest BS and the second closest BS is larger than $2/3$, i.e., $\left\| {\bf{d}}_0\right\| / \left\|{\bf{d}}_1 \right\| > 2/3$. 
For obtaining the edge users sum throughput, we only pick the active edge users and sum their throughput. 
In the simulation, SNR is set to be $100{\rm dB}$. Fig.~\ref{sum_throughput_compare} illustrates the edge users sum throughput. As observed in the figure, when the dropped users are dense enough, the performance of the proposed method achieves the performance of the dynamic clustering. Specifically, the proposed method have the same performance as the dynamic clustering when $K_{\rm perBS}>20$ in the $p=100$ case, and when $K_{\rm perBS}>40$ in the $p=200$ case. Comparing to the performance of the static clustering in the dense user environment, the edge users throughput is improved by $28\%$ in the $p=100$ case and $24\%$ in the $p=200$ case.
The rationale behind this improvement obviously comes from that users communicate their two closest BSs in the proposed method, where the strongest interference is mitigated. In the static clustering, the strongest interference still can degrade the performance. 
In contrary to that, when the dropped users are not dense enough, the performance of the proposed method is even lower than that of the static clustering. This is because the proposed method uses predefined cooperative regions (2nd-order Voronoi regions) and time-frequency allocation. For intuition, let's consider an example cooperative region $\mathcal{V}_2\left({\bf{d}}_i, {\bf{d}}_j \right)$. For serving users in this region, a BS cluster $\{ {\bf{d}}_i, {\bf{d}}_j\}$ is made and a specific time-frequency resource is allocated to this cluster. If there is no user in this region, however, the allocated time-frequency resource is wasted, and this degrades the performance. In the dynamic clustering, this resource waste does not occur since a BS cluster is flexibly made by a global scheduler, therefore there is always a user to communicate with the BS cluster. 
When the users are dense enough, all the predefined region has at least $K$ users therefore no resource is wasted. Then, the proposed method provides the same performance as the dynamic clustering. This intuition also justifies the difference between the results of the $p=100$ and $p=200$ cases. Since the BSs locations are more regular when $p=100$, each 2nd-order Voronoi region has similar area. When $p=200$, however, the BSs locations are highly irregular, therefore some 2nd-order Voronoi regions can be very small. Then, the probability that each region has at least $K$ users is reduced. This is the reason why in the $p=200$ case, more users are required to achieve the same performance as the dynamic clustering compared to the $p=100$ case. 

\begin{figure}[!t] 
\centering
$\begin{array}{cc}  
{\resizebox{0.465\columnwidth}{!}
{\includegraphics{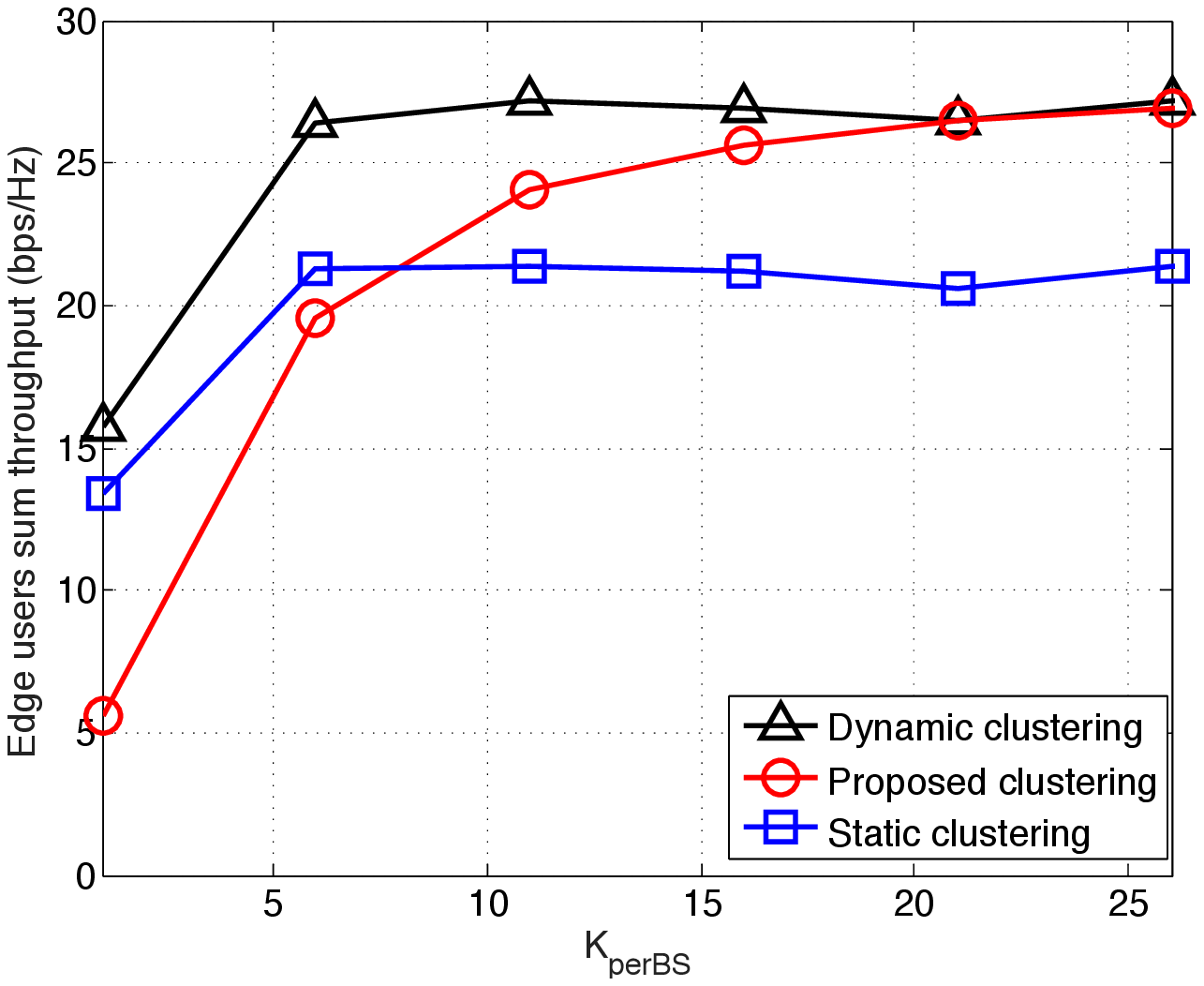}}}  &
{\resizebox{0.47\columnwidth}{!} 
{\includegraphics{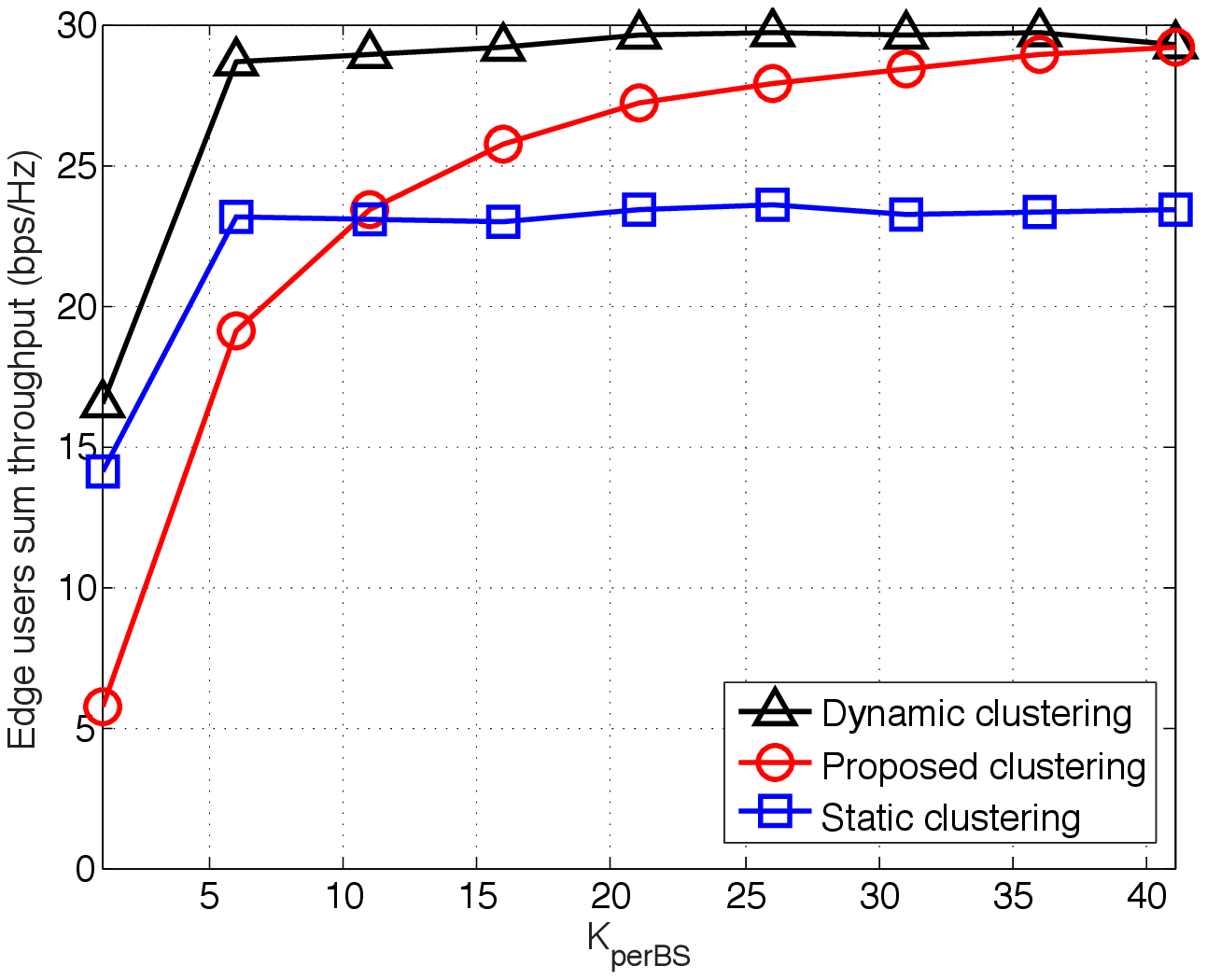}}}  \\   
\mbox{(a)} &
\mbox{(b)}
\end{array}$
\caption{The edge users sum throughput depending on $K_{\rm perBS}$ in the $p=100$ case (a) and in the $p=200$ case (b). It is assumed that $N=3$, $K=1$, and $\beta = 4$. The edge cutting is performed with $\Delta_{\rm EC}(G)=4$.}
\label{sum_throughput_compare}  
\end{figure} 

From this observation, we make some intuitive explanations about the effect of $\Delta_{\rm EC}(G)$ on the network performance. Considering $\Delta_{\rm EC}(G) = 1$ as an extreme case, no time-frequency resource are wasted. 
A large user density, however, would be required to achieve the same level of the coordination benefit with that of dynamic clustering. This is mainly because many 2nd-order Voronoi regions are cut, therefore users exist in the remaining 2nd-order Voronoi region with only small probability. In another extreme case, if the edge-cutting algorithm is not applied, i.e., $\Delta_{\rm EC}(G) = \Delta(G)$, the performance of the proposed method is not affected by the user density since all regions are covered by the proposed method, while the network performance degradation is unavoidable due to the resource waste. In summary, less user density is required as $\Delta_{\rm EC}(G)$ increases, while more resource waste is expected.

We can incorporate the channel estimation overhead into the sum throughput performance. From \cite{NY:dynamic}, the pilot overhead is 
\begin{align}
\alpha = \frac{L_p(2, N, {\rm SINR})}{L_b},
\end{align}
where $L_p(2, N, {\rm SINR})$ is the number of symbols for pilots and $L_b$ is a given fading coherence time in symbols. The number of symbols for pilots is $L_p(2, N, {\rm SINR}) = 2\eta N$, and $\eta$ is 
\begin{align}
\eta = \max\left\{1, \left\lfloor \frac{1}{\rm SINR} \left(\frac{1}{\rm MMSE } -1\right)\right\rfloor \right\},
\end{align}
where ${\rm MMSE}$ is the minimum mean square error of the pilot signals. Incorporating the pilot overhead into the throughput, the performance is degraded by $(1-\alpha)$. 
Since all the considered coordination methods are pair-wise BS coordination where a BS cluster consists of two BSs, the amount of channel estimation overhead is same.


\begin{figure}[!t] 
\centering
$\begin{array}{cc}  
{\resizebox{0.47\columnwidth}{!}
{\includegraphics{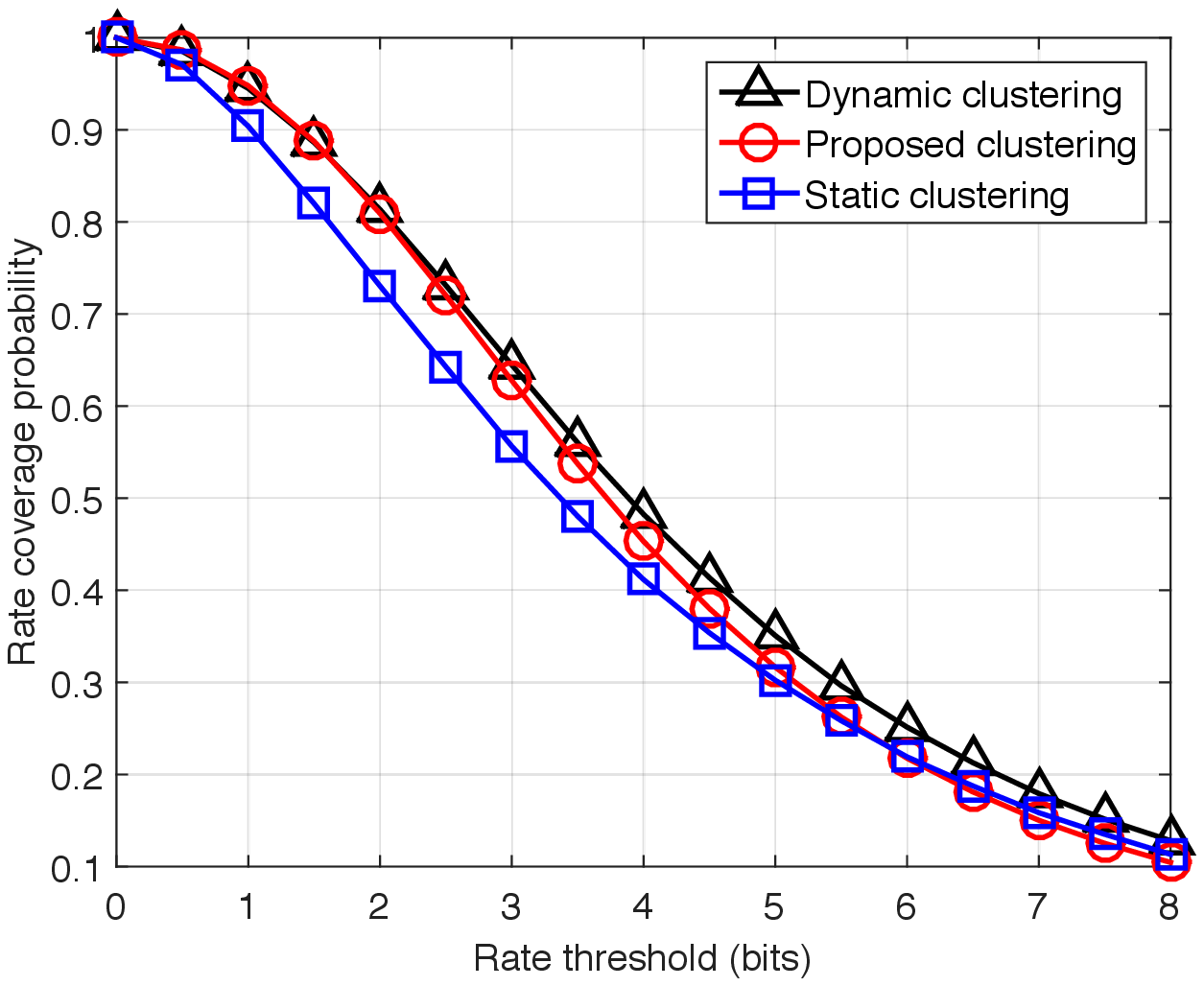}}} &
{\resizebox{0.455\columnwidth}{!}
{\includegraphics{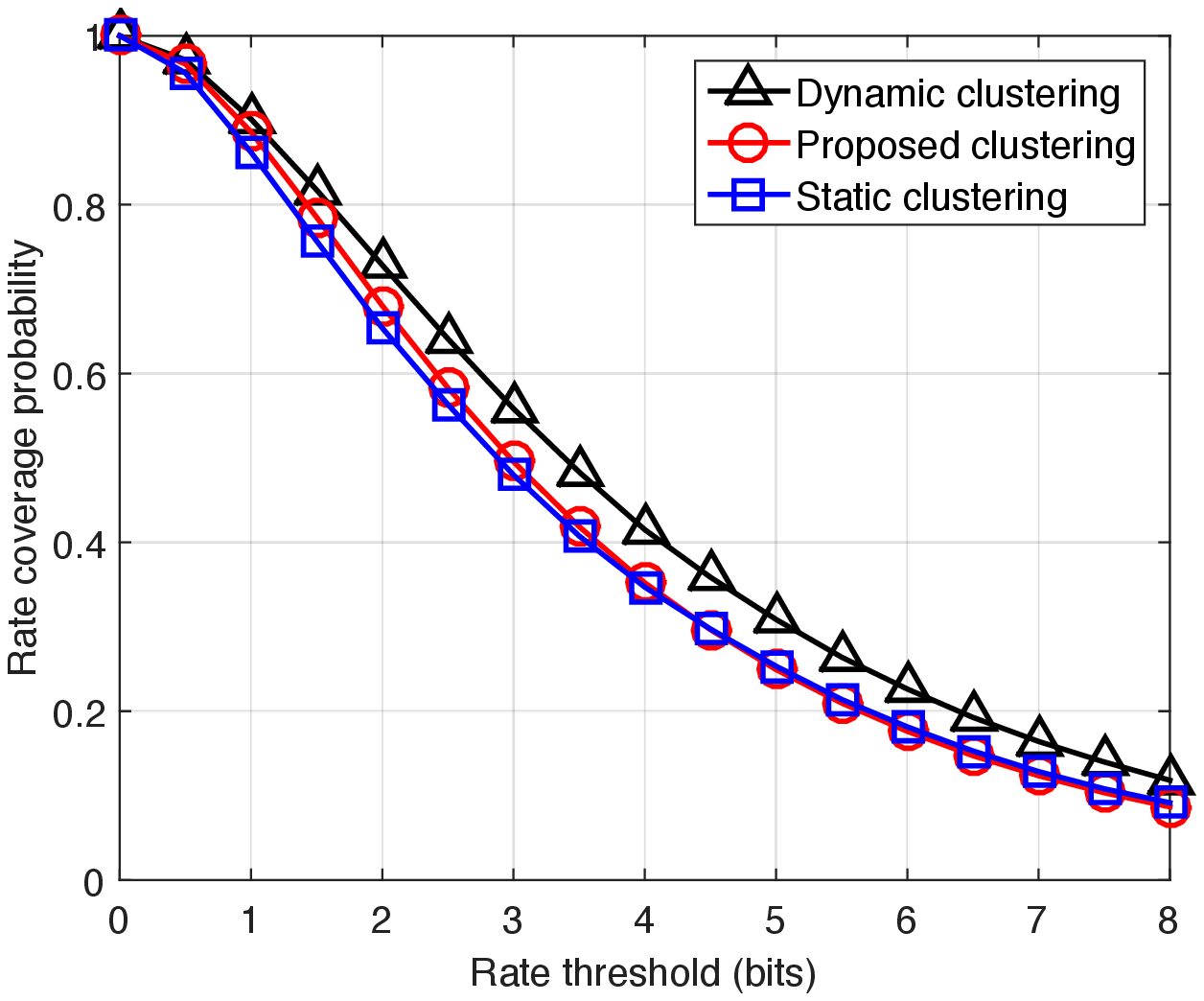}}}  \\  
\mbox{(a)} &
\mbox{(b)}
\end{array}$
\caption{The rate coverage probability in the $p=100$ case (a) and the $p=200$ case (b). It is assumed that $N=3$, $K=1$, and $\beta = 4$. The edge cutting is performed with $\Delta_{\rm EC}(G)=4$.}
\label{rate_compare_fig}  
\end{figure}

Now we compare the rate coverage probability. For obtaining the rate coverage probability, we calculate the probability that the selected user's throughput is larger than rate threshold $\gamma$. It is assumed that $\rm {SNR}=100 {\rm dB}$ and $K_{\rm perBS} = 40$.
Fig.~\ref{rate_compare_fig} illustrates the rate coverage probability. As observed in this figure, the proposed method provides the same or similar coverage probability with that of the dynamic clustering at low rate threshold, while the coverage probability degrades to the similar level of the static clustering at high rate threshold. Focusing on $\gamma = 1 {\rm bit}$, the rate outage probability is reduced by $50\%$ comparing to the static clustering in the $p=100$ case, and it is reduced by $15\%$ in the $p=200$ case. Since the rate coverage performance at low rate threshold is determined by edge users, the simulation results show that the proposed method is able to provide the same or similar level of coordination benefit with the dynamic clustering to edge users.
At high rate threshold, however, the rate coverage performance is degraded. The reason of this observation is on the edge cutting algorithm. As a result of the edge cutting algorithm, we cut some regions and do not serve users in the cut regions even if users' locations in the cut region are very close to their associated BS. Then, we miss a chance to have a high rate, which degrades the rate coverage at high rate threshold.

\section{Conclusion}
In this paper, we proposed a method for designing BS clusters and BS cluster patterns in 
irregular BS topologies. The core idea of the proposed method is to form BS clusters according to the 2nd-order Voronoi regions, and use edge-coloring for the graph induced by Delaunay triangulation to assign time-frequency resources (colors) to each BS cluster.
By using the proposed method, users in the uncut regions communicate with their two closest BS while avoiding the BS selection conflict problem.
In a fixed network with irregular BSs' locations, we provided analytical expressions of the 
rate coverage probability and the ergodic spectral efficiency in terms of the number of antennas per BS, the number of users, distance, SNR, and the pathloss exponent. 
In a random network modeled by a homogeneous PPP,
we also derived a lower bound of the ergodic spectral efficiency as a function of the pathloss exponent and the number of associated users.
Through system level simulations, we showed that the proposed method achieves the same edge users sum throughput performance with dynamic clustering when the users are dense enough. In the rate coverage probability performance, the proposed method provides the same or similar performance with that of dynamic clustering at low rate threshold.
Future work should consider the application of more advanced coordination strategies going beyond coordinated beamforming.

\bibliographystyle{IEEEtran}
\bibliography{referece_cmcbf}

\end{document}